\documentclass[runningheads]{llncs}

\usepackage{amsmath}
\usepackage{amssymb}
\newcommand{\Qed}{\qed}
\usepackage{latexsym}
\usepackage{enumerate}
\usepackage[normalem]{ulem}
\usepackage{color}
\usepackage{cite}

\usepackage{graphicx}
\usepackage{tikz}
\usetikzlibrary{calc,positioning,matrix,fit}

\usepackage[compatibility=false]{caption}
\usepackage[labelformat=parens]{subcaption}

\usepackage{algorithm}
\usepackage[noend]{algpseudocode}

\usepackage{color}

\usepackage{hyperref}
\hypersetup{
  colorlinks=true,       %
  linkcolor=blue!50!black
}

\newboolean{appendix}
\setboolean{appendix}{true}
\spnewtheorem*{lemma*}{Lemma (restated)}{\bfseries}{\itshape}
\spnewtheorem*{theorem*}{Theorem (restated)}{\bfseries}{\itshape}

\newcommand{\IC}{\textsf{IC}}
\newcommand{\NIC}{\textsf{NIC}}
\newcommand{\RAC}{\textsf{RAC}}
\newcommand{\One}{\textsf{1}}

\newcommand{\Wilog}{W.\,l.\,o.\,g.}
\newcommand{\wilog}{w.\,l.\,o.\,g.}
\newcommand{\ie}{i.\,e.}
\newcommand{\eg}{e.\,g.}
\newcommand{\ea}{et~al.}
\newcommand{\cf}{cf.}

\newcommand{\NP}{\ensuremath{\mathcal{NP}}}
\newcommand{\bigO}{\ensuremath{\mathcal{O}}}

\newcommand{\Edge}[2]{\{#1,#2\}}
\newcommand{\EdgeSegmentFirst}[2]{\{\underline{#1},#2\}}
\newcommand{\EdgeSegmentLast}[2]{\{#1,\underline{#2}\}}

\newcommand{\Drawing}[1]{\mathcal{D}(#1)}
\newcommand{\Emb}{\mathcal{E}}
\newcommand{\EmbOther}{\mathcal{E'}}
\newcommand{\Embedding}[1]{\Emb{}(#1)}
\newcommand{\EmbeddingOther}[1]{\EmbOther{}(#1)}

\newcommand{\Dual}[1]{{#1}^*}
\newcommand{\DualOther}[1]{{#1}^{*'}}
\newcommand{\PSub}[1]{\widehat{#1}}
\newcommand{\Skel}[1]{\widetilde{#1}}
\newcommand{\Trivial}{trivial}
\newcommand{\Kitonic}{marked}
\newcommand{\Free}{unmarked}
\newcommand{\Tetrahedral}{tetrahedral}
\newcommand{\Maximum}{densest}

\newcommand{\FigLabel}[1]{\label{fig:#1}}
\newcommand{\TabLabel}[1]{\label{tab:#1}}

\newcommand{\LemLabel}[1]{\label{lem:#1}}
\newcommand{\ThmLabel}[1]{\label{thm:#1}}
\newcommand{\CorLabel}[1]{\label{cor:#1}}

\newcommand{\ConjLabel}[1]{\label{conj:#1}}

\newcommand{\Fig}[1]{Fig.\,\ref{fig:#1}}
\newcommand{\FigAnd}[2]{Fig.\,\ref{fig:#1} and \ref{fig:#2}}
\newcommand{\Tab}[1]{Table~\ref{tab:#1}}

\newcommand{\Lem}[1]{Lemma~\ref{lem:#1}}
\newcommand{\Thm}[1]{Theorem~\ref{thm:#1}}
\newcommand{\Cor}[1]{Corollary~\ref{cor:#1}}
\newcommand{\Alg}[1]{Algorithm~\ref{alg:#1}}

\newcommand{\SectLabel}[1]{\label{sect:#1}}
\newcommand{\Sect}[1]{Sect.\,\ref{sect:#1}}

\newcommand{\KiteFace}{\raisebox{-1pt}{\includegraphics{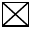}}}
\newcommand{\TriangleFace}{\raisebox{-1pt}{\includegraphics{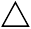}}}
\newcommand{\TetrahedronFace}{\raisebox{-1pt}{\includegraphics{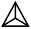}}}
\newcommand{\Rep}{\mathcal{P}}
\newcommand{\AdjacentFace}{\!\mid\!}

\newcommand{\DualNode}{q}
\newcommand{\DualNodeOther}{r}
\newcommand{\DualNodeThird}{s}
\newcommand{\DualNodeFourth}{t}

\newcommand{\TriangleFaceFrac}[2]{\ensuremath{\frac{#1\scalebox{.7}{\text{\TriangleFace}}}{#2}}}
\newcommand{\TetrahedronFaceFrac}[2]{\ensuremath{\frac{#1\scalebox{.7}{\text{\TetrahedronFace}}}{#2}}}

\newcommand{\Partition}{P}
\newcommand{\Sphere}{\mathcal{S}}
\newcommand{\QuarterSphere}{\Sphere_{Q}}
\newcommand{\Level}{\mathcal{L}}

\newcommand{\Kthree}{\ensuremath{\tau}}
\newcommand{\Triangle}{\Kthree}
\newcommand{\Kfour}{\ensuremath{\kappa}}
\newcommand{\KfourSet}{\ensuremath{\mathcal{K}}}
\newcommand{\KiteSet}{\ensuremath{\KfourSet_{\times}}}
\newcommand{\Kfive}{\ensuremath{\pi}}

\makeatletter
\providecommand*{\cupdot}{%
  \mathbin{%
    \mathpalette\@cupdot{}%
  }%
}
\newcommand*{\@cupdot}[2]{%
  \ooalign{%
    $\m@th#1\cup$\cr
    \hidewidth$\m@th#1\cdot$\hidewidth
  }%
}
\makeatother

\begin{document}

\title{%
\NIC{}-Planar Graphs %
\thanks{%
Supported by the Deutsche Forschungsgemeinschaft (DFG), grant Br835/18-1.}%
}

\author{%
Christian Bachmaier \and
Franz J.\,Brandenburg  \and
Kathrin Hanauer \and\\
Daniel Neuwirth \and
Josef Reislhuber}

\institute{%
University of Passau,
94030 Passau, Germany \\
\email{\{bachmaier|brandenb|hanauer|neuwirth|reislhuber\}@fim.uni-passau.de}
}

\authorrunning{C.\ Bachmaier, F.\ J.\ Brandenburg, K.\ Hanauer, D.\
  Neuwirth, J.\ Reislhuber}

\maketitle

\begin{abstract}
  A graph is \NIC{}-planar if it admits a drawing in the plane with at most
  one crossing per edge and such that two pairs of crossing edges share at
  most one common end vertex. \NIC{}-planarity generalizes \IC{}-planarity,
  which allows a vertex to be incident to at most one crossing edge, and
  specializes 1-planarity, which only requires at most one crossing per
  edge.

  We characterize embeddings of maximal \NIC{}-planar graphs in terms of
  generalized planar dual graphs. The characterization is used to derive
  tight bounds on the density of maximal \NIC{}-planar graphs which ranges
  between $3.2(n - 2)$ and $3.6(n - 2)$. Further, we show that optimal
  \NIC{}-planar graphs with $3.6(n - 2)$ edges have a unique embedding and
  can be recognized in linear
  time, whereas the recognition problem of \NIC{}-planar graphs is
  \NP-complete. In addition, we show that there are \NIC{}-planar graphs
  that do not admit right angle crossing drawings, which distinguishes
  \NIC{}-planar from \IC{}-planar graphs.
\end{abstract}

\section{Introduction}
Beyond-planar graphs, a family of graph classes defined as
extensions of planar graphs with different restrictions on
crossings, have received recent interest \cite{l-beyond-14}.
\One{}-planar graphs constitute an important class of this family. A
graph is \emph{\One{}-planar} if it can be drawn in the plane with
at most one crossing per edge. These graphs were introduced by
Ringel~\cite{ringel-65} in the context of coloring a planar graph
and its dual simultaneously and have been studied intensively since
then. Ringel observed that a pair of crossing edges can be augmented
by planar edges to form $K_4$. Bodendiek~\ea{}~\cite{bsw-bs-83,
bsw-1og-84} proved that \One{}-planar graphs with $n$ vertices have
at most $4n - 8$ edges, which is a tight bound for $n = 8$ and all
$n \geq 10$. These facts have also been discovered in other works.
\One{}-planar graphs with $4n - 8$ edges are called optimal
\cite{bsw-1og-84, s-s1pg-86}. They have a special structure and
consist of a triconnected planar quadrangulation with a pair of crossing edges in
each face. A graph $G$ is maximal 1-planar if no further edge can be
added to $G$ without violating 1-planarity.
Brandenburg~\ea{}~\cite{begghr-odm1p-13} found sparse maximal
1-planar graphs with less than $2.65n$
edges, which implies that there are maximal \One{}-planar graphs
that are not optimal and that are even sparser than maximal planar
graphs. The best known lower bound on the density of maximal
\One{}-planar graphs is $2.22 n$ \cite{bt-idm1p-15} and neither the
upper nor the lower bound is known to be tight.

There are some important subclasses of \One{}-planar graphs. A graph
is \emph{\IC{}-planar} (independent crossing planar)
\cite{a-cbirc-08, ks-cpgIC-10, zl-spgic-13, bdeklm-IC-16} if it has
a \One{}-planar embedding so that each vertex is incident to at most
one crossing edge. \IC{}-planar graphs were introduced by
Albertson~\cite{a-cbirc-08} who studied the coloring problem.
Kr{\'{a}}l and Stacho~\cite{ks-cpgIC-10} solved the coloring problem
and proved that $K_5$ is the largest complete graph that is
\IC{}-planar. \IC{}-planar graphs have an upper bound of $3.25n - 6$
on the number of edges, which is known as a tight bound as there are optimal
\IC{}-planar graphs with $13k - 6$ edges for all $n = 4k$ and $k \geq
2$ \cite{zl-spgic-13}. For other values of $n$ with $n \geq 8$ the
maximum number of edges is $\lfloor 3.25n-6 \rfloor$. On the other
hand, there are sparse maximal \IC{}-planar graphs with only $3n -
5$ edges for all $n \geq 5$~\cite{bbh-nipg-17}. In \emph{\NIC{}-planar} graphs
(near-independent crossing planar) \cite{z-dcmgprc-14}, two pairs of
crossing edges share at most one vertex. Equivalently, if every pair
of crossing edges is augmented by planar edges to $K_4$, an edge may
be part of at most one $K_4$ that is embedded with a crossing. Note that a
graph is \NIC{}-planar if every biconnected component is \NIC{}-planar.
\NIC{}-planar graphs were introduced by Zhang~\cite{z-dcmgprc-14},
who proved a density of at most $3.6(n - 2)$ and showed that $K_6$
is not \NIC{}-planar. Czap and {\v S}ugarek \cite{cs-tc1pg-14} give
an example of an optimal \NIC{}-planar graph with 27 vertices and 90
edges which proves that the upper bound is tight.
\emph{Outer \One{}-planar} graphs are another subclass of \One{}-planar
graphs. They must admit a \One{}-planar embedding such that all vertices
are in the outer face~\cite{abbghnr-o1p-15, heklss-ltao1p-15}. Results on
the density of maximal graphs are summarized in \Tab{density}.

\begin{table}[tb]
\caption{The density of maximal graphs. An asterisk marks our results.}
\TabLabel{density}
\setlength{\tabcolsep}{3pt}
\renewcommand{\arraystretch}{1.2}
\centering
\begin{tabular}{r|c|c|c|c}
\noalign{\smallskip}
& 1-planar
& \NIC{}-planar
& \IC{}-planar
& outer 1-planar
\\ \hline
upper bound
& $4n - 8$~\cite{bsw-bs-83,bsw-1og-84}
& $\frac{18}{5}(n - 2)$~\cite{z-dcmgprc-14},~(*) %
& $\frac{13}{4}n - 6$~\cite{ks-cpgIC-10}
& $\frac{5}{2}n - 2$~\cite{abbghnr-o1p-15}
\\
$\lfloor$ example & $4n-8$~\cite{bsw-bs-83,bsw-1og-84} &
$\frac{18}{5}(n - 2)$~\cite{cs-tc1pg-14},~(*) &
$\frac{13}{4}n - 6$~\cite{zl-spgic-13} & $\frac{5}{2}n -
2$~\cite{abbghnr-o1p-15}
\\
lower bound
& $\frac{20}{9} n - \frac{10}{3}$~\cite{bt-idm1p-15}
& $\frac{16}{5}(n-2)$~(*)
& $3n-5$~\cite{bbh-nipg-17}
& $\frac{11}{5}n - \frac{18}{5}$~\cite{abbghnr-o1p-15}
\\
$\lfloor$ example
& $\frac{45}{17}n - \frac{84}{17}$~\cite{begghr-odm1p-13}
& $\frac{16}{5}(n - 2)$~(*)
& $3n-5$~\cite{bbh-nipg-17}
& $\frac{11}{5}n - \frac{18}{5}$~\cite{abbghnr-o1p-15}
\\
\end{tabular}
\end{table}

There is a notable interrelationship between \One{}-planar and
\RAC{} graphs which are graphs that can be drawn straight-line with
right angle crossings~\cite{del-dgrac-11, el-racg1p-13}.
Didimo \ea{}~\cite{del-dgrac-11} showed that \RAC{} graphs have at
most $4n - 10$ edges and proved that there are \emph{optimal} \RAC{}
graphs with $4n - 10$ edges for all $n = 3k + 5$ and $k \geq 3$. For
other values of $n$ it is unknown whether there are optimal \RAC{}
graphs. Eades and Liotta \cite{el-racg1p-13} established that
optimal \RAC{} graphs (they called them maximally dense) admit a
special structure and proved that optimal \RAC{} graphs are
\One{}-planar. However, not all \RAC{}-graphs are \One{}-planar and vice
versa \cite{el-racg1p-13}, \ie, the classes of \RAC{}-graphs and
\One{}-planar graphs are incomparable.
Recently, Brandenburg~\ea{}~\cite{bdeklm-IC-16} showed that
every \IC{}-planar graph admits a \RAC{} drawing, which implies that every
\IC{}-planar graph is a \RAC{} graph.
They posed the problem whether \NIC{}-planar graphs are \RAC{} graphs,
which we refute.
Hence, with respect to \RAC{} drawings, \NIC{}-planar graphs behave like
\One{}-planar graphs and differ from \IC{}-planar graphs.

Recognizing \One{}-planar graphs is \NP-complete in general
\cite{GB-AGEFCE-07, km-mo1ih-13}, and remains \NP-complete even for
graphs of bounded bandwidth, pathwidth, or
treewidth~\cite{bce-pc1p-13}, if an edge is added to a planar
graph~\cite{cm-aoepl-13}, and if the input graph is triconnected and
given with a rotation system~\cite{abgr-1prs-15}. Likewise, testing
\IC{}-planarity is \NP-complete~\cite{bdeklm-IC-16}. On the other hand,
there are polynomial-time recognition algorithms for \One{}-planar
graphs that are maximized in some sense, such as
triangulated~\cite{cgp-rh4mg-06}, maximal~\cite{b-4mg1p-15}, and
optimal graphs~\cite{b-optlin-16}.

In this paper we study \NIC{}-planar graphs. After some basic
definitions in \Sect{basics}, we characterize embeddings of maximal
\NIC{}-planar graphs in terms of generalized planar dual graphs in
\Sect{dual}, and derive tight upper and lower bounds on the density
of maximal \One{}-planar graphs in \Sect{density} for infinitely
many values of $n$. A linear-time recognition algorithm for optimal
\NIC{}-planar graphs is established in \Sect{densest-recognition}.
In \Sect{nic-is-not-rac} we show that \NIC{}-planar graphs are
incomparable with \RAC{} graphs, and we consider the recognition problem
for \NIC{}-planar graphs.
We conclude in \Sect{conclusion} with some open problems.

\section{Preliminaries}%
\label{sect:basics}
We consider simple, undirected graphs $G = (V, E)$ with $n$ vertices and $m$
edges and assume that the graphs are biconnected.
The subgraph induced by a subset $U \subseteq V$ of vertices is denoted by
$G[U]$.

A \emph{drawing} $\Drawing{G}$ is a mapping of $G$ into the plane
such that the vertices are mapped to distinct points and each edge
is mapped to a Jordan arc connecting the points of the end vertices.
Two edges \emph{cross} if their Jordan arcs intersect in a point
different from their end points.
Crossings subdivide an edge into two or more uncrossed \emph{edge segments}.
An edge without crossings is called \emph{planar} and consists only of a
\emph{\Trivial{}} edge segment.
Edge segments of crossed edges are said to be \emph{non-\Trivial{}}.
For a crossed edge $\Edge{u}{v}$, we denote
the extremal edge segment that is incident to $u$ ($v$) by
$\EdgeSegmentFirst{u}{v}$ ($\EdgeSegmentLast{u}{v}$).
A drawing is \emph{planar} if every edge is planar, and
\emph{\One{}-planar} if there is at most one crossing per edge.

A drawing of a graph partitions the plane into empty regions called
\emph{faces}.
A face is defined by the cyclic sequence of edge segments that forms
its \emph{boundary}, which is described by vertices and crossing points,
\eg, face $f_{ab}$ in \Fig{kite-embedding}.
Two faces $f$ and $g$ are said to be \emph{adjacent}, denoted as $f \AdjacentFace g$,
if their boundaries share a common edge segment.
A vertex $v$ is \emph{incident} to a face $f$ if there is an edge $\Edge{u}{v}$
such that $f$ is either bounded by the \Trivial{} edge segment $\Edge{u}{v}$ or
the non-\Trivial{} edge segment $\EdgeSegmentLast{u}{v}$.
A face is called a (\emph{\Trivial{}}) \emph{triangle}, if its boundary consists
of exactly three (\Trivial{}) edge segments.
The set of all faces describes the \emph{embedding} $\Embedding{G}$.
The embedding resulting from a planar (\One{}-planar) drawing is called
\emph{planar} (\emph{\One{}-planar}).
A \One{}-planar embedding coincides with the embedding of the
\emph{planarization} of $G$ which is obtained by treating each crossing
point as a vertex and the edge segments as edges.
A \One{}-planar embedding is \emph{triangulated} if every face
is a triangle~\cite{cgp-rh4mg-06}.
A \emph{planar reduction} $\PSub{G}_{\Emb{}} \subseteq G$
is a planar subgraph of $G$ obtained by
removing exactly one edge of each pair of edges that cross in $\Embedding{G}$.
The \emph{planar skeleton} $\Skel{G}_{\Emb{}} \subseteq G$ is the planar
subgraph of $G$ with respect to $\Embedding{G}$ that is obtained by
removing all pairs of crossing edges.

We consider \NIC{}-planar graphs and embeddings that are maximized in
some sense.
A \NIC{}-planar embedding $\Embedding{G}$ is \emph{maximal (planar-maximal)
\NIC{}-planar} if no further (planar) edge can be added to $\Embedding{G}$
without violating \NIC{}-planarity or simplicity. A \NIC{}-planar graph $G$
is \emph{maximal} if the graph $G+e$ obtained from $G$ by the addition of an
edge $e$ is no longer \NIC{}-planar. A graph is called a \emph{sparsest}
(\emph{densest}) \NIC{}-planar graph if it is maximal \NIC{}-planar with
the minimum (maximum) number of edges among all maximal \NIC{}-planar graphs
of the same size. A \NIC{}-planar graph $G$ is called \emph{optimal} if $G$
has exactly the upper bound of $3.6 (n - 2)$ edges.
There are analogous definitions for other graph classes.

The concepts planar-maximal embedding and maximal and optimal
graphs coincide for planar graphs, where the maximum number of edges is
always $3n - 6$. However, they are different for \One{}-planar, \NIC{}-planar,
and \IC{}-planar graphs. The complete graph on five vertices without one
edge, $K_5-e$, is planar and has a planar embedding which is planar-maximal.
Nevertheless, $K_5-e$ can be embedded with a pair of crossing edges and $e$ can
be added as a planar edge. A graph is maximal if every embedding is maximal.
All the same, an embedding $\Embedding{G}$ of a graph $G$ may be maximal
\One{}-planar (\NIC{}-planar, \IC{}-planar) without $G$ being maximal. As
mentioned before, there are sparse \One{}-planar graphs with less than $2.65 n$
edges, whereas optimal 1-planar graphs have $4n - 8$ edges~\cite{s-rm1pg-10}.
Due to integrality, optimal \NIC{}-planar graphs exist only
for $n = 5k + 2$ and optimal \IC{}-planar graphs only for $n = 4k$.
Zhang and Liu~\cite{zl-spgic-13} present optimal \IC{}-planar graphs
with $4k$ vertices and $13k - 6$ edges for every $k \geq 2$ vertices and Czap
and {\v S}ugarek~\cite{cs-tc1pg-14} gave an example of an optimal
\NIC{}-planar graph of size 27. We show that there are optimal
\NIC{}-planar graphs for all $n = 5k+2$ with $k \geq 2$ and not for
$n=7$. The distinction between densest and optimal graphs is
important, since optimal graphs often have a special structure, as we
shall show for \NIC{}-planar graphs.

\begin{figure}[tb]
\centering
\begin{tikzpicture}
\node[inner sep=0pt,draw=none] (l) {
  \includegraphics[scale=.8]{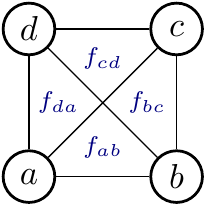}
  \phantomsubcaption\label{fig:kite-embedding}};
\node[inner sep=0pt,draw=none] (m) [right=of l] {
  \includegraphics[scale=.8]{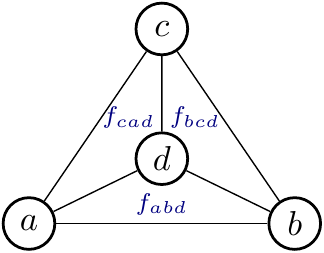}
  \phantomsubcaption\label{fig:simpletetrahedron-embedding}};
\node[inner sep=0pt,draw=none] (r) [base right=of m] {
  \includegraphics[scale=.8]{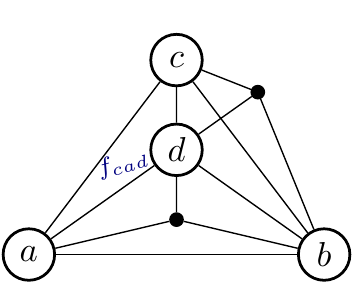}
  \phantomsubcaption\label{fig:nonsimpletetrahedron-embedding}};
\path (l.north west) |- node[anchor=north] {(\subref*{fig:kite-embedding})} (r.north west);
\path (m.north west) |- node[anchor=north] {(\subref*{fig:simpletetrahedron-embedding})} (r.north west);
\node[anchor=north] at (r.north west) {(\subref*{fig:nonsimpletetrahedron-embedding})};
\end{tikzpicture}
\caption{%
The two embeddings of a $K_4$ induced by the vertices $\{a, b, c, d\}$ (up to isomorphism):
A kite~(\subref{fig:kite-embedding})
and a tetrahedron, which can be simple~(\subref{fig:simpletetrahedron-embedding})
or non-simple~(\subref{fig:nonsimpletetrahedron-embedding}).}%
\label{fig:k4-embeddings}
\end{figure}
The complete graph on four vertices $K_4$ plays a crucial role in 1-planar
(\IC{}- and \NIC{}-planar) graphs.
It has exactly the two \One{}-planar embeddings depicted in
\Fig{k4-embeddings}~\cite{Kyncl-09}.
If $K_4$ is a subgraph of another graph $G$, further vertices and
edges of $G$ may be inside the shown faces.
Let $\Embedding{G}$ be a \NIC{}-planar embedding of a graph $G = (V,E)$
and let $U = \{a, b, c, d\} \subseteq V$
such that $G[U]$ is $K_4$.
Denote by $\Embedding{G[U]}$ the embedding of $G[U]$
induced by $\Embedding{G}$.
$G[U]$ is embedded as a \emph{kite} in $\Embedding{G}$
(see also \Fig{kite-embedding}) if, \wilog,
$\Edge{a}{c}$ and $\Edge{b}{d}$ cross each other and
there are faces $f_{ab}, f_{bc}, f_{cd}, f_{da}$ in $\Embedding{G}$
such that $f_{ab}$
is bounded exactly by $\Edge{a}{b}$, $\EdgeSegmentFirst{a}{c}$, and
$\EdgeSegmentFirst{b}{d}$,
and analogously for $f_{bc}, f_{cd}, f_{da}$.
Hence, there is no other vertex in the interior of a kite.
$G[U]$ is embedded as a \emph{tetrahedron} in $\Embedding{G}$ if
all edges are planar with respect to $\Embedding{G[U]}$
but not necessarily in $\Embedding{G}$.
The tetrahedron embedding of $G[U]$ in $\Embedding{G}$ is called \emph{simple},
if, \wilog, $d$ has vertex degree three and $f_{abd}$, $f_{bcd}$,
and $f_{cad}$ are faces in $\Embedding{G}$.
Then $d$ is called the \emph{center} of the tetrahedron.
\Fig{simpletetrahedron-embedding} shows a simple tetrahedron embedding
of $G[U]$,
whereas the tetrahedron embedding in \Fig{nonsimpletetrahedron-embedding} is
non-simple due to the missing faces $f_{abd}$ as well as $f_{bcd}$.

\section{The Generalized Dual of Maximal \NIC{}-planar Graphs}
\SectLabel{dual}
In this section, we study the structure of $\NIC{}$-planar embeddings of
maximal \NIC{}-planar graphs.
We use the results for tight upper and lower bounds of the density of
\NIC{}-planar graphs in \Sect{density} and for a linear-time recognition
algorithm for optimal \NIC{}-planar graphs in \Sect{densest-recognition}.
The upper bound on the density was proved in~\cite{z-dcmgprc-14} using a
different technique and there is a construction for the tightness of the
bound in~\cite{cs-tc1pg-14} in steps of 15 vertices.

Let $\Embedding{G}$ be a \NIC{}-planar embedding of a maximal \NIC{}-planar
graph $G$. The first property we observe is that the subgraph of $G$ induced
by the end vertices of a pair of crossing edges induces $K_4$ with a kite
embedding.
It has been discovered by Ringel~\cite{ringel-65} and in many other works that
a pair of crossing edges admits such a \One{}-planar embedding.
This fact is used for a normal form of embeddings of triconnected 1-planar
graphs~\cite{abk-sld3c-13} in the sense that it can always be established.
With respect to maximal \NIC{}-planarity, every embedding obeys this normal
form.
\begin{figure}[tb]
\centering
\begin{tikzpicture}
\node[inner sep=0pt,draw=none] (l) {
  \includegraphics[scale=.9]{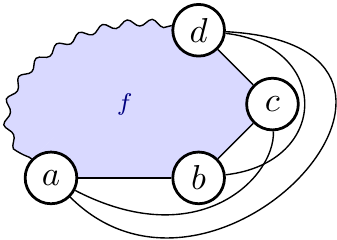}%
  \phantomsubcaption\label{fig:ttriangle-a}};
\node[anchor=center] at (l.north west) {(\subref*{fig:ttriangle-a})};
\node[inner sep=0pt,draw=none,anchor=west] (r) at ($(l.east) + (1.3cm,0)$) {
  \includegraphics[scale=.9]{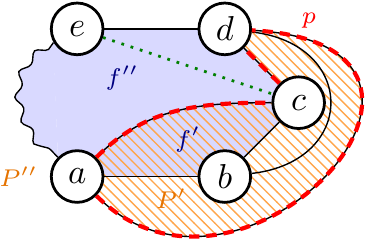}%
  \phantomsubcaption\label{fig:ttriangle-b}};
\node[anchor=center] at (r.north west) {(\subref*{fig:ttriangle-b})};
\end{tikzpicture}
\caption{%
Proof of \Lem{kite-or-triangle}.%
}%
\label{fig:proof-ttriangle}
\label{fig:proof-adjfaces}
\end{figure}
\begin{lemma}
\LemLabel{crossing-implies-kite}
Let $\Edge{a}{c}$, $\Edge{b}{d}$ be two edges crossing each other in a
\NIC{}-planar embedding of a maximal \NIC{}-planar graph $G = (V,E)$.\\
Then, $\Edge{a}{b}, \Edge{b}{c}, \Edge{c}{d}, \Edge{a}{d} \in E$ and
the induced $K_4$ is embedded as a kite.
\end{lemma}
\begin{proof}
Let $\Embedding{G}$ be a \NIC{}-planar embedding of $G$.
Consider a pair of vertices $e \in \{ \{a,b\}, \{b, c\}, \{ c, d\}, \{d, a\}\}$.
As $G$ is maximal, $e \in E$, otherwise it could be added without violating
\NIC{}-planarity.
Thus $G[\{a,b,c,d\}]$ is $K_4$.

\Wilog, let $e = \Edge{a}{b}$ (the other cases are similar).
Due to the crossing of $\Edge{a}{c}$ and $\Edge{b}{d}$,
the vertices $a$ and $b$ cannot be incident to another pair of crossing edges.
Hence, $e$ must be planar in $\Embedding{G}$.
Let $f$ denote the face of $\Embedding{G}$ that is incident to both $a$ and $b$
and bounded by the edge segments $\EdgeSegmentFirst{a}{c}$ and
$\EdgeSegmentFirst{b}{d}$.
Suppose that $f$ is not bounded by $\Edge{a}{b}$.
Furthermore, $e$ and the edge segments $\EdgeSegmentFirst{a}{c}$
and $\EdgeSegmentFirst{b}{d}$ form a closed
path $p$ that partitions the set of faces of $\Embedding{G}$.
Let $g$ and $h$ denote the faces bounded by $\Edge{a}{b}$ in $\Embedding{G}$.
The fact that $a$ and $b$ cannot be incident to another pair of crossing edges
implies that there is a vertex $x \neq a, b$ incident to $g$ as well as a
vertex $y \neq a,b$ incident to $h$.
As $p$ places $g$ and $h$ in different partitions and contains neither $x$ nor $y$,
$x \neq y$.
Moreover, $p$ consists of one planar edge and two edge segments and
thus cannot be crossed.
Subsequently, $x$ and $y$ are not adjacent.

Construct a new embedding $\EmbeddingOther{G}$ from $\Embedding{G}$ by re-routing
$\Edge{a}{b}$ such that it subdivides $f$.
Then, $g$ and $h$ conflate into one face $gh$ in $\EmbeddingOther{G}$.
As $x$ and $y$ are both incident to $gh$ in $\EmbeddingOther{G}$, the edge
$\Edge{x}{y}$ can be added to $G$ and embedded planarly such that it subdivides
$gh$.
The resulting embedding is \NIC{}-planar, thus contradicting the maximality of $G$.
\Qed{} %
\end{proof}

\Lem{crossing-implies-kite} implies that every face bounded by at least
one non-\Trivial{} edge segment is part of a kite and hence a
non-\Trivial{} triangle.
In fact, every embedding of a maximal \NIC{}-planar graph is triangulated.

\begin{lemma}
\LemLabel{kite-or-triangle}
Every face of a \NIC{}-planar embedding $\Embedding{G}$ of a maximal
\NIC{}-planar graph $G$ with $n \geq 5$ is a triangle.
\end{lemma}
\begin{proof}
Let $\Embedding{G}$ be a \NIC{}-planar embedding of $G = (V,E)$ and
let $f$ be a face of $\Embedding{G}$.
Recall from \Sect{basics} that a triangle is either \Trivial{} if its
boundary consists only of \Trivial{} edge segments, \ie, planar edges, or
non-\Trivial{} otherwise.

By \Lem{crossing-implies-kite}, every face whose boundary contains
non-\Trivial{} edge segments is part of a kite.
Hence, if $f$ is bounded by at least one non-\Trivial{} edge segment,
it is a non-\Trivial{} triangle.

Otherwise, assume that $f$ is bounded only by %
planar edges
and suppose that $f$ is not a triangle.
Let $a,b,c,d \in V$ be distinct vertices incident to $f$ such that $f$'s
boundary contains the edges $\Edge{a}{b},\Edge{b}{c},\Edge{c}{d}$.
Suppose that $G$ does not contain edge $\Edge{a}{c}$. Then $\Edge{a}{c}$
could be added to $G$ and embedded planarly such that it subdivides $f$. The
resulting embedding would be \NIC{}-planar, a contradiction to the
maximality of $G$. Thus, $\Edge{a}{c} \in E$, but it is not part of the
boundary of $f$.
By the same reasoning, if $\Edge{b}{d} \not\in E$ then it could be added.
Hence, $\Edge{a}{c},\Edge{b}{d} \in E$.
As neither of both edges subdivides $f$, $\Edge{a}{c}$ and $\Edge{b}{d}$ must
cross each other in $\Embedding{G}$ and thus, by \Lem{crossing-implies-kite},
$\Embedding{G}$ contains a corresponding kite of the $K_4$ induced by
$a,b,c,d$.
In particular, $\Edge{a}{d} \in E$ (\cf~\Fig{ttriangle-a}).
Note that $\Edge{a}{d}$ does not necessarily bound $f$.
Nevertheless, this kite can be transformed into a tetrahedron
by re-routing one of $\Edge{a}{c}$ or $\Edge{b}{d}$ such that it
subdivides $f$.

\Wilog, consider the embedding $\EmbeddingOther{G}$ obtained from
$\Embedding{G}$ by re-routing $\Edge{a}{c}$.
Denote by $f'$ and $f''$ the faces resulting from the division of $f$
and let $f'$ be the face bounded by $\Edge{a}{b},\Edge{b}{c},\Edge{a}{c}$.
In $\EmbeddingOther{G}$, the edges $\Edge{a}{c}, \Edge{c}{d}, \Edge{a}{d}$ form
a closed planar path $p$ and thus split the set of faces into two partitions
$\Partition'$ and $\Partition''$, such that $\Partition'$ contains exactly the
two triangles emerging from the conflation of the former kite faces as
well as the face $f'$ (\cf~\Fig{ttriangle-b}).
Note that only the vertices $a,b,c,d$ are incident to faces in $\Partition'$.
If $n \geq 5$, there must be a fifth vertex, which is subsequently incident
only to faces contained in $\Partition''$.
In particular, this vertex must be incident to at least two faces in
$\Partition''$ due to biconnectivity.
This implies that $\Partition''$ cannot consist only of
face $f''$, and consequently, $f''$ cannot be bounded
by the edge $\Edge{a}{d}$.
Recall that $f''$ results from the subdivision of face $f$ in $\Embedding{G}$
and that $f$'s boundary consists only of planar edges.
Hence, there is a vertex $e$ incident to $f''$ such that $f''$
is bounded by $\Edge{d}{e}$ and $e$ is distinct from $a,b,c,d$.
Observe that $p$ contains $c$, but not $e$, and, by construction, $f''$ is
the only face $c$ is incident to in $\Partition''$.
As $p$ is planar and the boundary of $f''$ contains $\Edge{a}{c}$,
$\Edge{c}{d}$, and $\Edge{d}{e}$, $c$ and $e$ cannot be adjacent in $G$.
However, $\Edge{c}{e}$ can be added to $G$ and embedded planarly such that it
subdivides $f''$, which yields a \NIC{}-planar embedding, a contradiction to
the maximality of $G$.
\Qed{}
\end{proof}
Observe that the lemma neither holds for maximal \One{}-planar graphs nor for
plane-maximal \One{}-planar graphs, as both admit hermits, \ie, vertices of
degree two~\cite{begghr-odm1p-13}.
Moreover, the restriction to graphs with $n \geq 5$ is indispensable,
as $K_4$ can be embedded as a kite, which has a quadrangular (outer) face.
Forthcoming, we assume that all graphs have size $n \geq 5$.

Next, consider the planar versions of an embedding $\Embedding{G}$
of a maximal \NIC{}-planar graph $G$.
\begin{corollary}
\CorLabel{psub-maxplanar}
\CorLabel{triconnected}
\CorLabel{triconnected-skeleton}
Let $G$ be a maximal \NIC{}-planar graph with $n \geq 5$.
\begin{itemize}
\vspace*{-.2cm}
\item
If $\PSub{G}_{\Emb{}} \subseteq G$ is a planar reduction of $G$ with respect
to any \NIC{}-planar embedding $\Embedding{G}$, then $\PSub{G}_{\Emb{}}$ is
maximal planar.
\item $G$ is triconnected.
\item
The planar skeleton $\Skel{G}_{\Emb{}} \subseteq G$
with respect to every \NIC{}-planar embedding $\Embedding{G}$
is triconnected.
\end{itemize}
\end{corollary}
\begin{proof}
$\PSub{G}_{\Emb{}}$ is planar and triangulated
by \Lem{kite-or-triangle}.
Thus, it is triconnected and so is $G$ as its supergraph.
Then, also the planar skeleton $\Skel{G}_{\Emb{}}$ is triconnected
as shown by Alam \ea~\cite{abk-sld3c-13}.
\Qed{}
\end{proof}
These results enable us to define the \emph{generalized dual graph}
of a maximal \NIC{}-planar graph.
As in the case of planar graphs, the dual is defined with respect to a specific
embedding $\Embedding{G}$.
Figure~\ref{fig:dual-example} provides a small example.
\begin{figure}[tb]
\centering
\begin{tikzpicture}[node distance=.3cm]
\matrix[matrix of nodes, column sep=.3cm,row sep=0.3cm,nodes={anchor=center}] (m) {%
\includegraphics{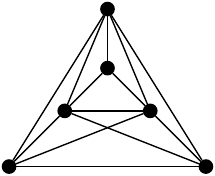}
\phantomsubcaption\label{fig:dual-example-graph}
&
\includegraphics{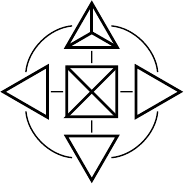}
\phantomsubcaption\label{fig:dual-example-dual}
\\
};
\node[fit=(m-1-1) (m-1-2)] (row1) {};
\path (m-1-1.south west) |- node[anchor=north west] {(\subref*{fig:dual-example-graph})} (row1.north east);
\path (m-1-2.south west) |- node[anchor=north west] {(\subref*{fig:dual-example-dual})} (row1.north east);
\end{tikzpicture}
\caption{An embedding of a graph~(\subref{fig:dual-example-graph}) and the
corresponding generalized dual~(\subref{fig:dual-example-dual}).}%
\FigLabel{dual-example}
\end{figure}
\begin{definition}
The \emph{generalized dual graph}
$\Dual{G} = (\Dual{V},\Dual{E})$ of a maximal \NIC{}-planar graph $G$
with respect to \NIC{}-planar embedding $\Embedding{G}$
contains three types of nodes:
For every set of faces forming a kite in $\Embedding{G}$, $\Dual{V}$ contains a
\KiteFace{}-node.
For every set of faces forming a simple tetrahedron in $\Embedding{G}$,
$\Dual{V}$ contains a \TetrahedronFace{}-node.
All other faces of $\Embedding{G}$ are represented by a \TriangleFace{}-node
each.
Let $\DualNode \in \Dual{V}$ and denote by
$\Rep(\DualNode)$ the set of faces of $\Embedding{G}$ represented by
$\DualNode$.
As in conventional duals of planar graphs, there is
an edge $\Edge{\DualNode}{\DualNodeOther} \in \Dual{E}$ for every pair of
adjacent faces $f \AdjacentFace g$ of $\Embedding{G}$ such that $f \in
\Rep(\DualNode)$ and $g \in \Rep(\DualNodeOther)$ and $\DualNode \neq
\DualNodeOther$.
\end{definition}
In the following, we analyze the structure of $\Dual{G}$.
For clarification, we call the elements of $V$ vertices and the
elements of $\Dual{V}$ nodes.
Note that \KiteFace{}-nodes have degree four and
\TetrahedronFace{}- and \TriangleFace{}-nodes have degree three.
Furthermore, this definition in general allows for multi-edges, but not loops.
\begin{lemma}
\LemLabel{dual-triconnected}
The generalized dual graph $\Dual{G}$ of a maximal \NIC{}-planar
graph $G$ with respect to a \NIC{}-planar embedding $\Embedding{G}$
is a simple $3$-connected planar graph.
\end{lemma}
\begin{proof}
Let $G' \subseteq G$ be the graph obtained from $G$ by
removing all vertices of degree three in $G$.
Recall that every face in $\Embedding{G}$ is a triangle by
\Lem{kite-or-triangle}.
Let $\EmbeddingOther{G'}$ denote the \NIC{}-planar
embedding of $G'$ inherited from $\Embedding{G}$
and observe that $\EmbeddingOther{G'}$ emerges from
$\Embedding{G}$ by
removing the center vertex of each simple tetrahedron and
thereby replacing every set of faces
forming a simple tetrahedron embedding by a \Trivial{}
triangle.
Next, consider the planar skeleton $\Skel{G'}_{\EmbOther}$ of $G'$ with
respect to $\EmbeddingOther{G'}$, which is obtained by removing all pairs of
crossing edges in kites. Then the generalized dual graph $\Dual{G}$ of $G$
with respect to $\Embedding{G}$ is the planar dual of
$\Skel{G'}_{\EmbOther}$.
As $\Skel{G'}_{\EmbOther}$ is simple and triconnected by
\Cor{triconnected-skeleton}, so is its dual graph.
\Qed
\end{proof}
\begin{figure}[tb]
\centering
\begin{tikzpicture}[node distance=.3cm]
\matrix[matrix of nodes, column sep=.3cm,row sep=0.3cm,nodes={anchor=center}] (m) {%
\includegraphics{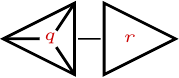}
\phantomsubcaption\label{fig:proof-dualadj-triangle-tetrahedron-dual}
&
\includegraphics[scale=.7]{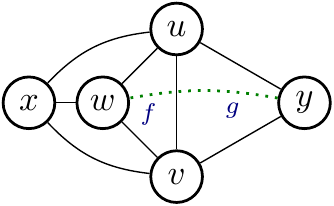}
\phantomsubcaption\label{fig:proof-dualadj-triangle-tetrahedron}
&
\includegraphics{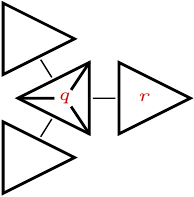}
\phantomsubcaption\label{fig:proof-dualadj-tetrahedron-triangles-dual}
\\
\includegraphics[scale=.7]{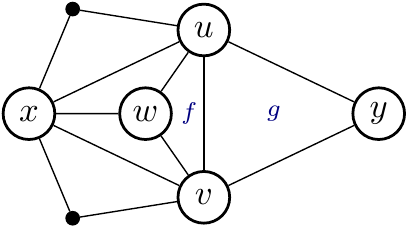}
\phantomsubcaption\label{fig:proof-dualadj-tetrahedron-triangles}
&
\includegraphics[scale=.7]{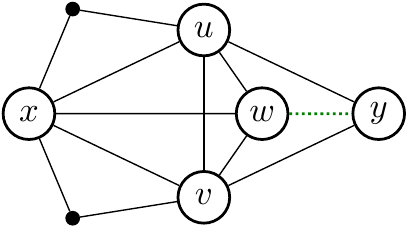}
\phantomsubcaption\label{fig:proof-dualadj-tetrahedron-triangles-reembed}
&
\includegraphics[scale=.7]{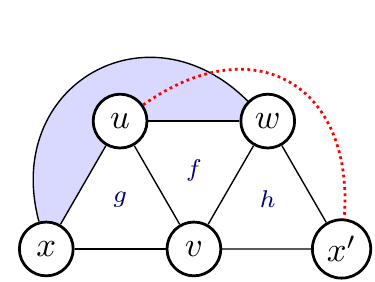}
\phantomsubcaption\label{fig:proof-dualadj-triangle-triangle}
\\
};
\node[fit=(m-1-1) (m-1-2) (m-1-3)] (row1) {};
\node[fit=(m-2-1) (m-2-2) (m-2-3)] (row2) {};
\node[fit=(m-1-1) (m-2-1)] (col1) {};
\node[fit=(m-1-2) (m-2-2)] (col2) {};
\node[fit=(m-1-3) (m-2-3)] (col3) {};

\path (col1.south west) |- node[anchor=north west] {(\subref*{fig:proof-dualadj-triangle-tetrahedron-dual})} (row1.north east);
\path (col2.south west) |- node[anchor=north west] {(\subref*{fig:proof-dualadj-triangle-tetrahedron})} (row1.north east);
\path (col3.south west) |- node[anchor=north west] {(\subref*{fig:proof-dualadj-tetrahedron-triangles-dual})} (row1.north east);
\path (col1.south west) |- node[anchor=north west] {(\subref*{fig:proof-dualadj-tetrahedron-triangles})} (row2.north east);
\path (col2.south west) |- node[anchor=north west] {(\subref*{fig:proof-dualadj-tetrahedron-triangles-reembed})} (row2.north east);
\path (col3.south west) |- node[anchor=north west] {(\subref*{fig:proof-dualadj-triangle-triangle})} (row2.north east);
\end{tikzpicture}
\caption{%
Proof of \Lem{dual-adjacencies}:\\
Case~\ref{itm:dual-adj-tetra}:
A \protect\TetrahedronFace{}-node $\DualNode$ in $\Dual{G}$ that is adjacent to an \Free{} \protect\TriangleFace{}-node $\DualNodeOther$
(\subref{fig:proof-dualadj-triangle-tetrahedron-dual})
and the corresponding situation in $\Embedding{G}$, where $f \in \Rep(\DualNode)$ and $\Rep(\DualNodeOther) = \{ g\}$
(\subref{fig:proof-dualadj-triangle-tetrahedron}).
\\
Case~(\ref{itm:dual-tetra-marked}):
A \protect\TetrahedronFace{}-node $\DualNode$ in $\Dual{G}$ that is adjacent to a \protect\TriangleFace{}-node $\DualNodeOther$
as well as two further \protect\TriangleFace{}-nodes
(\subref{fig:proof-dualadj-tetrahedron-triangles-dual}),
the corresponding situation in $\Embedding{G}$, where $f \in \Rep(\DualNode)$ and $\Rep(\DualNodeOther) = \{ g\}$
(\subref{fig:proof-dualadj-tetrahedron-triangles}),
and the reembedding, which then allows for the insertion of an additional edge
(\subref{fig:proof-dualadj-tetrahedron-triangles-reembed}).
\\
Case (\ref{itm:dual-adj-triangles-tetra}):
The situation in $\Embedding{G}$ in case of two adjacent, \Free{} \protect\TriangleFace{}-nodes
$\DualNode, \DualNodeOther$
in $\Dual{G}$ with $\Rep(\DualNode) = \{ f \}$ and $\Rep(\DualNodeOther) = \{ g \}$
(\subref{fig:proof-dualadj-triangle-triangle}).
The shaded region contains further vertices and edges of the graph.%
}%
\label{fig:proof-dual-adjacencies-tetrahedron}%
\label{fig:proof-dual-adjacencies}
\end{figure}
We say that a \TetrahedronFace{}- or \TriangleFace{}-node $\DualNode \in \Dual{V}$ is
\emph{\Kitonic{}} if $\DualNode$ is adjacent to a \KiteFace{}-node in $\Dual{G}$.
Otherwise, $\DualNode$ is \emph{\Free{}}.
Two adjacent \TriangleFace{}-nodes $\DualNode, \DualNodeOther \in \Dual{V}$ are said to be
\emph{\Tetrahedral{}} if the union of their vertices induces $K_4$ in $G$.
Observe that this induced $K_4$ is necessarily embedded as a non-simple
tetrahedron.
Let $u, v, w$ and $u, v, x$ denote the vertices incident to the faces
represented by $\DualNode$ and $\DualNodeOther$, respectively.
Then, $\Edge{w}{x} \in E$ if and only if $\DualNode, \DualNodeOther$ are
\Tetrahedral{}.
In this case, we call $\Edge{w}{x}$ the \emph{\Tetrahedral{} edge}.
In \Fig{proof-dualadj-triangle-triangle}, \eg, the \TriangleFace{}-nodes
representing the faces $f$ and $g$ in $\Dual{G}$ are \Tetrahedral{} and
$\Edge{w}{x}$ is the corresponding \Tetrahedral{} edge.

The definition of a \NIC{}-planar embedding implies a number of restrictions on
the adjacencies among nodes in $\Dual{G}$, which are subsumed in the following
lemma.
\begin{lemma}
\LemLabel{dual-adjacencies}
Let $\Dual{G}$ be the generalized dual of a maximal \NIC{}-planar graph
$G$ where $n \geq 5$ with respect to a \NIC{}-planar embedding $\Embedding{G}$.
\begin{enumerate}[(i)]
\item%
\label{itm:dual-adj-kite}
No two \KiteFace{}-nodes are adjacent.
\item%
\label{itm:dual-adj-tetra}
A \TetrahedronFace{}-node is only adjacent to \KiteFace{}-nodes and \Kitonic{}
\TriangleFace{}-nodes.
\item%
\label{itm:dual-tetra-marked}
Every \TetrahedronFace{}-node is \Kitonic{}.
\item%
\label{itm:dual-adj-triangles-tetra}
If two \Free{} \TriangleFace{}-nodes are adjacent,
then they are \Tetrahedral{}.
\item%
\label{itm:dual-adj-triangles-free}
If a \TriangleFace{}-node is adjacent to two \TriangleFace{}-nodes, then
one of them is \Kitonic{}.
\end{enumerate}
\end{lemma}
\begin{proof}
Consider a node $\DualNode$ of $\Dual{G}$ that
is adjacent to another node $\DualNodeOther$ via an edge $e$.
Let $\Edge{u}{v}$ be the corresponding primal edge of $e$
and denote by $f \in \Rep(\DualNode)$ and $g \in \Rep(\DualNodeOther)$
the faces bounded by $\Edge{u}{v}$.
Note that $\Edge{u}{v}$ is planar in $\Embedding{G}$.

(\ref{itm:dual-adj-kite})
Assume that $\DualNode$ and $\DualNodeOther$ are \KiteFace{}-nodes.
Then, $u$ and $v$ would be incident to
two pairs of crossing edges, thus contradicting the \NIC{}-planarity of
$\Embedding{G}$.

(\ref{itm:dual-adj-tetra})
Next, assume that $\DualNode$ is a \TetrahedronFace{}-node.
Let $w, x \neq u, v$ denote the two other vertices of the $K_4$
represented by $\DualNode$, such that $w$ is the vertex of degree
three in the center of the tetrahedron.\\
Suppose that $\DualNodeOther$ is also a \TetrahedronFace{}-node.
As $f$ and $g$ are \Trivial{} triangles, $u$ and $v$ are not incident to a
common pair of crossing edges by \Lem{crossing-implies-kite}.
Let $y$ denote the third vertex incident to $g$, which is
the center of the tetrahedron embedding represented by $\DualNodeOther$.
Then, $w$ and $y$ are not adjacent to each other, and no pair of vertices from
$u, v, w, y$ is incident to a common pair of crossing edges.
In consequence, inserting an edge $\Edge{w}{y}$
and embedding it such that it crosses $\Edge{u}{v}$ does not violate
\NIC{}-planarity and thus contradicts the maximality of $G$.
Subsequently, $\DualNodeOther$ cannot be a \TetrahedronFace{}-node.\\
Suppose that $\DualNodeOther$ is a \TriangleFace{}-node
(\cf~\FigAnd{proof-dualadj-triangle-tetrahedron-dual}{proof-dualadj-triangle-tetrahedron}).
By \Lem{crossing-implies-kite}, $u$ and $v$ cannot be adjacent to a common pair
of crossing edges, as both $f$ and $g$ are \Trivial{} triangles.
Let $y \neq u,v$ be the third vertex incident to $g$.
If $x = y$, then $\Dual{G}$ would consist of exactly one \TriangleFace{}-node
and one \TetrahedronFace{}-node, \ie, $G$ is $K_4$ and $\Embedding{G}$ is its
planar embedding.
As $n \geq 5$, $x \neq y$.
Being the tetrahedron's center vertex, $w$ is adjacent only to $u$, $v$, and
$x$ and all edges incident to $w$ are planar.
Suppose that $\DualNodeOther$ is \Free{}.
Then, by \Lem{crossing-implies-kite}, neither $u,y$ nor $v,y$ are
adjacent to a common pair of crossing edges.
Subsequently, the edge $\Edge{w}{y}$ can be added to $G$ and embedded
such that it crosses $\Edge{u}{v}$ without violating \NIC{}-planarity,
a contradiction to $G$ being maximal.
Hence, if a \TetrahedronFace{}-node is adjacent to a \TriangleFace{}-node,
the latter must be \Kitonic{}.\\
(\ref{itm:dual-tetra-marked})
Suppose that $\DualNode$ is a \TetrahedronFace{}-node as in
(\ref{itm:dual-adj-tetra}) and only adjacent to (\Kitonic{})
\TriangleFace{}-nodes
(\cf~\FigAnd{proof-dualadj-tetrahedron-triangles-dual}{proof-dualadj-tetrahedron-triangles}).
Then, in particular, $\DualNodeOther$ is a \TriangleFace{}-node.
Denote the third vertex incident to $g$ by $y$.
As $\Edge{u}{v}$, $\Edge{v}{x}$, and $\Edge{u}{x}$ are planar and each of them
bounds two \Trivial{} triangles,
none of their end vertices can be
incident to a common pair of crossing edges by \Lem{crossing-implies-kite}.
The same holds for $u,w$ as well as $v,w$ and $x,w$, since $w$ is incident only
planar edges.
Subsequently, the embedding $\EmbeddingOther{G}$ obtained from $\Embedding{G}$
by reembedding $w$ such that $\Edge{w}{x}$ crosses $\Edge{u}{v}$ yields a
\NIC{}-planar embedding of $G$
(\cf~\Fig{proof-dualadj-tetrahedron-triangles-reembed}).
Furthermore, $\EmbeddingOther{G}$ contains a face that is bounded by four
\Trivial{} edge segments $\Edge{u}{w}$, $\Edge{v}{w}$, $\Edge{v}{y}$, and
$\Edge{u}{y}$.
Thus, there also is a \NIC{}-planar embedding for the graph obtained from $G$
by adding an edge $\Edge{w}{y}$, a contradiction to $G$ being maximal.
Note that as $n \geq 5$, $x \neq y$, hence, $\Edge{w}{y}$ is not contained
in $G$.
In consequence, the \TetrahedronFace{}-node $\DualNode$ cannot be adjacent to
\TriangleFace{}-nodes only and thus must itself be \Kitonic{}.

(\ref{itm:dual-adj-triangles-tetra})
Consider now the case that $\DualNode$ is a \TriangleFace{}-node
and let $w$ denote the third vertex incident to $f$.
Furthermore, assume that $\DualNodeOther$ is another \TriangleFace{}-node and
both are \Free{} (\cf~\Fig{proof-dualadj-triangle-triangle}).
Denote the third vertex incident to $g$ by $x$.
In consequence of \Lem{crossing-implies-kite}, no pair of $u,v,w$ is
incident to a common pair of crossing edges, and likewise for
$u,v,x$.
Suppose that $\Edge{w}{x} \not\in E$.
Then, by \Lem{crossing-implies-kite}, $w$ and $x$ cannot be incident to a
common pair of crossing edges.
Hence, $\Edge{w}{x}$ can be added to $G$ and embedded such that it crosses
$\Edge{u}{v}$ without violating \NIC{}-planarity, a contradiction to the
maximality of $G$.
Subsequently, $\Edge{w}{x} \in E$.
This in turn implies that $G[u,v,w,x]$ is $K_4$, hence, $\DualNode$ and
\DualNodeOther{} are \Tetrahedral{}.

(\ref{itm:dual-adj-triangles-free})
Finally, consider the case that $\DualNode$ and $\DualNodeOther$ are \TriangleFace{}-nodes
and that $\DualNode$ is adjacent to a second \TriangleFace{}-node $\DualNodeThird$
representing a face $h$, which is also depicted in \Fig{proof-dualadj-triangle-triangle}.
\Wilog, let $\Edge{v}{w}$ be the edge bounding both $f$ and $h$ and denote by
$x' \neq v,w$ the third vertex incident to $h$.
Suppose that $\DualNode$, $\DualNodeOther$, and $\DualNodeThird$ are \Free{}.
As we have argued in the proof of (\ref{itm:dual-adj-triangles-tetra}),
$\Edge{w}{x} \in E$.
For $\DualNode$ and $\DualNodeThird$, this analogously implies that
$\Edge{u}{x'}$ $\in E$.
As $\Edge{w}{x}$, $\Edge{x}{v}$, and $\Edge{v}{w}$ form a closed path and both
$\Edge{w}{x}$ and $\Edge{x}{v}$ are planar, $\Edge{u}{x'}$ must cross
$\Edge{w}{x}$.
By \Lem{crossing-implies-kite}, $G[u,w,x,x']$ is $K_4$ and $\Edge{u}{w}$ hence
bounds a non-\Trivial{} triangle.
Thus, $\DualNode$, $\DualNodeOther$, and $\DualNodeThird$ are \Kitonic{},
a contradiction.
Consequently, $\DualNode$, $\DualNodeOther$, and $\DualNodeThird$
cannot be all \Free{}.
\Qed{}
\end{proof}
With respect to a fixed embedding, also the converse of
\Lem{dual-adjacencies} holds:

\begin{lemma}%
\LemLabel{dual-adj2}
Let $\Dual{G}$ be a triconnected planar graph with vertices of degree three that are
labeled by \TriangleFace{} and \TetrahedronFace{} and of vertices of degree four
that are labeled by \KiteFace{} so that the requirements of
\Lem{dual-adjacencies} hold.
Let $G$ be a \One{}-planar graph whose generalized dual graph is $\Dual{G}$.
Then, $G$ is simple and triconnected and has a maximal \NIC{}-planar embedding.
\end{lemma}
\begin{proof}
A simple triconnected planar graph $H$ has a simple triconnected dual
$H^*$ and $H$ is isomorphic to the dual of $H^*$. Both graphs have a
unique planar embedding.

Let $G_1$ be the planar dual of $G^*$. Then $G_1$ is simple and
triconnected and has a unique planar embedding. First, add a pair of
crossing edges in each quadrilateral face of $G_1$. This is the
expansion of each \KiteFace{}-node. It preserves simplicity since
otherwise $G_1$ were not triconnected. The so obtained graph $G_2$ is
\One{}-planar. Next, expand each \TetrahedronFace{}-node by inserting a
center vertex in the respective triangle and call the resulting graph
$G_3$. This preserves simplicity and triconnectivity.
The  embedding $\Embedding{G_3}$ is inherited from the embedding
of $G_1$  and is triangulated and \One{}-planar.
It is also \NIC{}-planar, since two kites are not adjacent by
requirement (\ref{itm:dual-adj-kite}) of \Lem{dual-adjacencies}.
It remains to show that $\Embedding{G_3}$ is maximal \NIC{}-planar.
Towards a contradiction, suppose an edge $\Edge{u}{v}$ could be
added to $\Embedding{G_3}$. Then $u$ and $v$ are in a face $f$ of
$\Embedding{G_1}$ or in two adjacent faces $f_1, f_2$ with a common
edge $\Edge{a}{b}$ where the faces may be expanded by a center.
Face  $f$ must be a quadrilateral, which, however, is expanded to a
kite so that $\Edge{u}{v}$  already exists in $G_3$. If $u$ is the
center of a \TetrahedronFace{}-node, then the given requirements
from \Lem{dual-adjacencies} exclude a new edge, since the other face
is \Kitonic{} and the new edge would introduce two adjacent kites.
If $u$ and $v$ are the vertices on opposite sides of two triangles
with a common edge $\Edge{a}{b}$, then $\Edge{u}{v}$ is excluded by
the requirements.
$G$ is isomorphic to $G_3$, since $G^*$ is obtained from each as a
generalized dual.
\Qed
\end{proof}

Let $\DualNode \in \Dual{V}$ be a \KiteFace{}-node in $\Dual{G}$ and
$\DualNodeOther, \DualNodeThird \in \Dual{V}$ two \TriangleFace{}-nodes
such that $\DualNodeOther, \DualNodeThird$ are \Tetrahedral{}
and the corresponding \Tetrahedral{} edge is one of the crossing edges
of the kite represented by $\DualNode$.
A \emph{kite flip} between $\DualNode$, $\DualNodeOther$, and $\DualNodeThird$ is a
reembedding of the \Tetrahedral{} edge such
that it crosses the edge on the common boundary of the faces represented by
$\DualNodeOther$ and $\DualNodeThird$.
In the generalized dual $\DualOther{G}$ with respect to this new embedding
$\EmbeddingOther{G}$, $\DualNode$ is hence replaced by a pair of adjacent,
\Tetrahedral{} \TriangleFace{}-nodes and $\DualNodeOther$ and $\DualNodeThird$
are replaced by a \KiteFace{}-node.
More formally, if $\DualNode$ is adjacent to nodes
$\DualNode_a$, $\DualNode_b$,
$\DualNode_c$, $\DualNode_d$,
$\DualNodeOther$ is adjacent to $\DualNodeThird$, $\DualNodeOther_a$, and $\DualNodeOther_b$,
and
$\DualNodeThird$ is additionally adjacent to $\DualNodeThird_a$, and $\DualNodeThird_b$,
then $\DualOther{G}$ has the same set of vertices and edges as $\Dual{G}$
except for $\DualNode$, $\DualNodeOther$, and $\DualNodeThird$ and their incident edges.
Instead, $\DualOther{G}$ contains two adjacent, \Tetrahedral{}
\TriangleFace{}-nodes $\DualNode'$ and $\DualNode''$ and a \KiteFace{}-node
$\DualNodeFourth_{\DualNodeOther\DualNodeThird}$ such that, \wilog,
$\DualNode'$ is adjacent to $\DualNode_a$ and $\DualNode_b$,
$\DualNode''$ is adjacent to $\DualNode_c$ and $\DualNode_d$,
and $\DualNodeFourth_{\DualNodeOther\DualNodeThird}$ is adjacent to $\DualNodeOther_a$, $\DualNodeOther_b$,
$\DualNodeThird_a$, and $\DualNodeThird_b$ in $\DualOther{G}$.
\Fig{kite-flipping} provides an example, where the \KiteFace{}-node
is even adjacent to the \TriangleFace{}-nodes it is flipped with,
which is, however, not a prerequisite.
\begin{figure}[tb]
\centering
\begin{tikzpicture}[node distance=.3cm]
\matrix[matrix of nodes, column sep=.4cm,row sep=0.3cm,nodes={anchor=center}] (m) {%
\includegraphics[scale=.7]{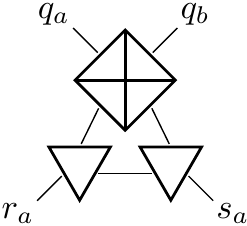}
\phantomsubcaption\label{fig:kite-flipping-top-dual}
&
\includegraphics[scale=.7]{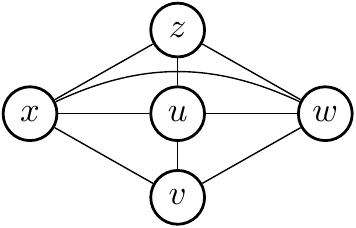}
\phantomsubcaption\label{fig:kite-flipping-top}
&
\includegraphics[scale=.7]{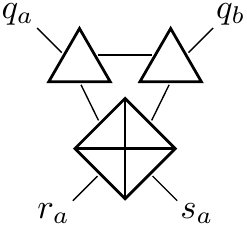}
\phantomsubcaption\label{fig:kite-flipping-bottom-dual}
&
\includegraphics[scale=.7]{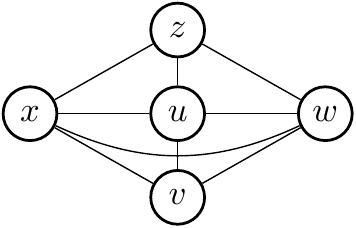}
\phantomsubcaption\label{fig:kite-flipping-bottom}
\\
};
\node[fit=(m-1-1) (m-1-2) (m-1-3) (m-1-4)] (row1) {};
\path (m-1-1.north west) |- node[anchor=north east] {(\subref*{fig:kite-flipping-top-dual})} (row1.north east);
\path (m-1-2.north west) |- node[anchor=north west] {(\subref*{fig:kite-flipping-top})} (row1.north east);
\path (m-1-3.north west) |- node[anchor=north east] {(\subref*{fig:kite-flipping-bottom-dual})} (row1.north east);
\path (m-1-4.north west) |- node[anchor=north west] {(\subref*{fig:kite-flipping-bottom})} (row1.north east);
\end{tikzpicture}
\caption{Example of a kite flip between a \protect\KiteFace{}-node $\DualNode{}$
and two \protect\TriangleFace{}-nodes $\DualNodeOther{}$ and $\DualNodeThird{}$, where
$\DualNodeOther{}$ and $\DualNodeThird{}$ are additionally adjacent to $\DualNode{}$.
The adjacencies in the generalized dual graph are shown before
(\subref{fig:kite-flipping-top-dual}) and after (\subref{fig:kite-flipping-bottom-dual}) the flip,
the corresponding embeddings are depicted in (\subref{fig:kite-flipping-top})
and (\subref{fig:kite-flipping-bottom}).%
}%
\label{fig:kite-flipping}
\end{figure}

In general, the embedding resulting from a kite flip in a (maximal) \NIC{}-planar
embedding is not necessarily also maximal or even \NIC{}-planar.
\begin{lemma}%
\LemLabel{kite-flip}
Let $\Dual{G}$ be the generalized dual of a maximal \NIC{}-planar graph $G =
(V, E)$ where $n \geq 5$ with respect to a \NIC{}-planar embedding
$\Embedding{G}$.
For every pair of adjacent \TriangleFace{}-nodes in $\Dual{G}$ that are either
\Free{} or adjacent to a single, common \KiteFace{}-node there is a maximal
\NIC{}-planar embedding $\EmbeddingOther{G}$ of $G$ such that in the
corresponding generalized dual, these two \TriangleFace{}-nodes are kite
flipped with a \KiteFace{}-node.
\end{lemma}
\begin{proof}
Consider a pair of \TriangleFace{}-nodes $\DualNode$
and $\DualNodeOther$ in $\Dual{G}$ that are adjacent to each other via an edge
$e$.
Let $\Edge{u}{v}$ be the corresponding primal edge of $e$ and denote by $f$ and
$g$ the \Trivial{} triangles they represent, \ie,
$\Rep(\DualNode) = \{f\}$ and $\Rep(\DualNodeOther) = \{g\}$.
Note that $f$ and $g$ are both bounded by $\Edge{u}{v}$ and let $w$ and $x$, $w
\neq x$, denote the third vertex incident to $f$ and $g$, respectively.

Assume first that $\DualNode$ and $\DualNodeOther$ are \Free{},
as obtained, \eg, in case of an embedding as in \Fig{proof-kiteflip}.
By \Lem{dual-adjacencies}, $\DualNode$ and $\DualNodeOther$ are \Tetrahedral{}.
Hence, $\Edge{w}{x} \in E$ and is the corresponding \Tetrahedral{} edge.
Suppose that $w$ and $x$ are not incident to a common pair of crossing edges.
By \Lem{crossing-implies-kite} and \Lem{kite-or-triangle}, $\Edge{w}{x}$ thus
bounds two \Trivial{} triangles.
Then, however, reembedding $\Edge{w}{x}$ such that it crosses $\Edge{u}{v}$
yields a \NIC{}-planar embedding with a face whose boundary consists of four
\Trivial{} edge segments, a contradiction to either \Lem{kite-or-triangle} or
the maximality of $G$.
Hence, $w$ and $x$ must be incident to a common pair of crossing edges, \ie,
$G$ contains a $K_4$ induced by $w,x$ and two further vertices $y,z$.
Let $\EmbeddingOther{G}$ be the \NIC{}-planar embedding obtained from
$\Embedding{G}$ by reembedding $\Edge{w}{x}$ such that it crosses
$\Edge{u}{v}$~(\Fig{proof-kiteflip-after}).
As $G$ is maximal, so must be $\EmbeddingOther{G}$.
The reembeddability of $\Edge{w}{x}$ together with \Lem{crossing-implies-kite}
implies that $\Edge{w}{x}$ is crossed in both $\Embedding{G}$ and
$\EmbeddingOther{G}$.
Consequently,
the generalized dual $\DualOther{G}$ with respect to $\EmbeddingOther{G}$
is obtained from $\Dual{G}$ by replacing $\DualNode{}$ and $\DualNodeOther$
with a \KiteFace{}-node and by replacing the \KiteFace{}-node representing
the kite $G[\{w,x,y,z\}]$ with two \TriangleFace{}-nodes, \ie,
$\DualOther{G}$ is obtained by a kite flip between this \KiteFace{}-node
and $\DualNode{}$ and $\DualNodeOther$.
Note that $y,z \neq u,v$, otherwise, $\DualNode$ and $\DualNodeOther$ were
adjacent to the \KiteFace{}-node representing the embedding of $G[\{w,x,y,z\}]$
and therefore \Kitonic{}.

Otherwise, assume that $\DualNode{}$ and $\DualNodeOther$ are adjacent
to a common \KiteFace{}-node $\DualNodeThird$ in $\Dual{G}$.
Then, the faces represented by $\DualNodeThird$ are incident to both $w$ and
$x$ as well as either $u$ or $v$.
\Wilog, assume the former and let $y$ denote the fourth vertex of the $K_4$
represented by $\DualNodeThird$, \ie, $\DualNodeThird$ represents the embedding
of the $K_4$ $G[u,w,x,y]$.
Figure~\ref{fig:proof-kiteflip-adj} shows an embedding that corresponds to
this situation in $\Dual{G}$.
Subsequently, $G$ contains the edge $\Edge{w}{x}$, which implies that
$\DualNode{}$ and $\DualNodeOther$ are \Tetrahedral{}.
As neither $\DualNode{}$ nor $\DualNodeOther$ is adjacent to a further
\KiteFace{}-node in $\Dual{G}$ by the requirements of the lemma, $\Edge{v}{x}$
and $\Edge{v}{w}$ do not bound a non-\Trivial{} triangle face.
Hence, neither $u, v$ nor $v, x$ nor $v, w$ are adjacent to a common
pair of crossing edges by \Lem{crossing-implies-kite}.
Furthermore, $u, w$ and $u, x$ are adjacent only to the pair of crossing edges
$\Edge{u}{y}$ and $\Edge{w}{x}$.
Let $\EmbeddingOther{G}$ be the \NIC{}-planar embedding obtained from
$\Embedding{G}$ by reembedding $\Edge{w}{x}$ such that it crosses
$\Edge{u}{v}$, \ie, by applying a kite flip to $\DualNodeThird$,
$\DualNode{}$ and $\DualNodeOther$~(\Fig{proof-kiteflip-adj-after}).
As $\Edge{u}{y}$ and $\Edge{w}{x}$ do not cross in $\EmbeddingOther{G}$, $u, w$
and $u, x$ are now again only adjacent to one common pair of crossing edges each,
which is in both cases $\Edge{u}{v}$ and $\Edge{w}{x}$.
Thus, $\EmbeddingOther{G}$ is a \NIC{}-planar embedding of $G$, and, as $G$ is
maximal, so is $\EmbeddingOther{G}$.
\Qed{}
\end{proof}
\begin{figure}[tb]
\centering
\begin{tikzpicture}[node distance=.3cm]
\matrix[matrix of nodes, column sep=.5cm,row sep=0.3cm,nodes={anchor=center}] (m) {%
\includegraphics[scale=.7]{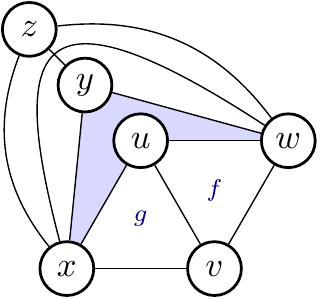}
\phantomsubcaption\label{fig:proof-kiteflip}
&
\includegraphics[scale=.7]{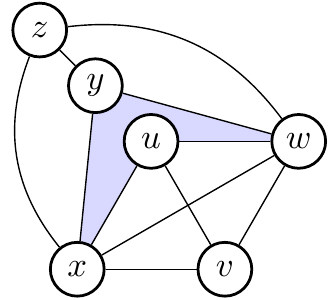}
\phantomsubcaption\label{fig:proof-kiteflip-after}
&
\includegraphics[scale=.7]{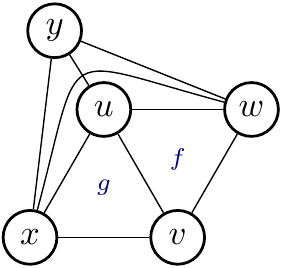}
\phantomsubcaption\label{fig:proof-kiteflip-adj}
&
\includegraphics[scale=.7]{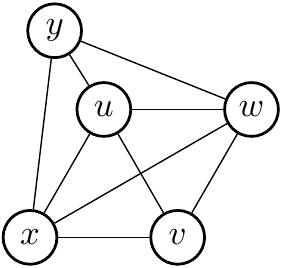}
\phantomsubcaption\label{fig:proof-kiteflip-adj-after}
\\
};
\node[fit=(m-1-1) (m-1-2) (m-1-3) (m-1-4)] (row1) {};
\path (m-1-1.north west) |- node[anchor=north east,inner sep=0pt] {(\subref*{fig:proof-kiteflip})} (row1.north east);
\path (m-1-2.north west) |- node[anchor=north east,inner sep=0pt] {(\subref*{fig:proof-kiteflip-after})} (row1.north east);
\path (m-1-3.north west) |- node[anchor=north east,inner sep=0pt] {(\subref*{fig:proof-kiteflip-adj})} (row1.north east);
\path (m-1-4.north west) |- node[anchor=north east,inner sep=0pt] {(\subref*{fig:proof-kiteflip-adj-after})} (row1.north east);
\end{tikzpicture}
\caption{%
Proof of \Lem{kite-flip}:
If the \protect\TriangleFace{}-nodes representing $f$ and $g$ are \Free{}~(\subref{fig:proof-kiteflip}),
the embedding obtained by a kite flip~(\subref{fig:proof-kiteflip-after}) is \NIC{}-planar.
The shaded region necessarily contains further vertices and edges.
Likewise, if the \protect\TriangleFace{}-nodes representing $f$ and $g$ are adjacent to a single,
common \protect\KiteFace{}-node~(\subref{fig:proof-kiteflip-adj}),
the embedding obtained by a kite flip~(\subref{fig:proof-kiteflip-adj-after}) is \NIC{}-planar.
}%
\FigLabel{fig:proof-kiteflips}
\end{figure}
The requirements of \Lem{dual-adjacencies} characterize a
\NIC{}-planar embedding $\Embedding{G}$ of a maximal \NIC{}-planar
graph $G$. They guarantee that $\Embedding{G}$ is maximal, but the
graph $G$ may still have another, non-maximal \NIC{}-planar embedding,
as \Fig{nic-maximality} shows.
Note that the second embedding emerges from a kite flip between the
\KiteFace{}-node representing the kite embedding of the $K_4$ induced by $a$,
$b$, and both red vertices and the two adjacent \TriangleFace{}-nodes
representing the \Trivial{} triangle faces that $e$, one red vertex and either
$b$ or the second red vertex are incident to.
\begin{figure}[tb]
  \centering
  \begin{tikzpicture}
    \node[inner sep=0pt, draw=none] (f1) {
      \includegraphics[scale=0.4]{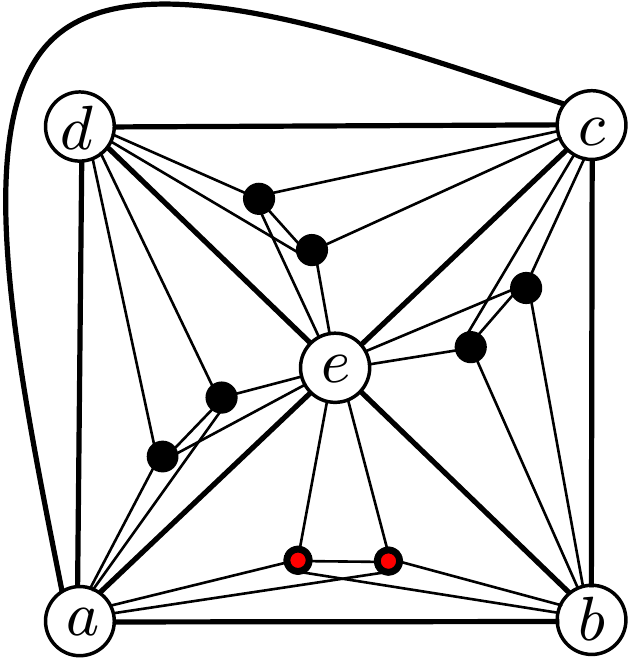}%
      \phantomsubcaption\FigLabel{nm-1}};
    \node[anchor=north east] at (f1.north west) {(\subref*{fig:nm-1})};
    \node[inner sep=0pt, draw=none, anchor=north west] (f2)
        at ($(f1.north east) + (1.5cm,0)$){
      \includegraphics[scale=0.4]{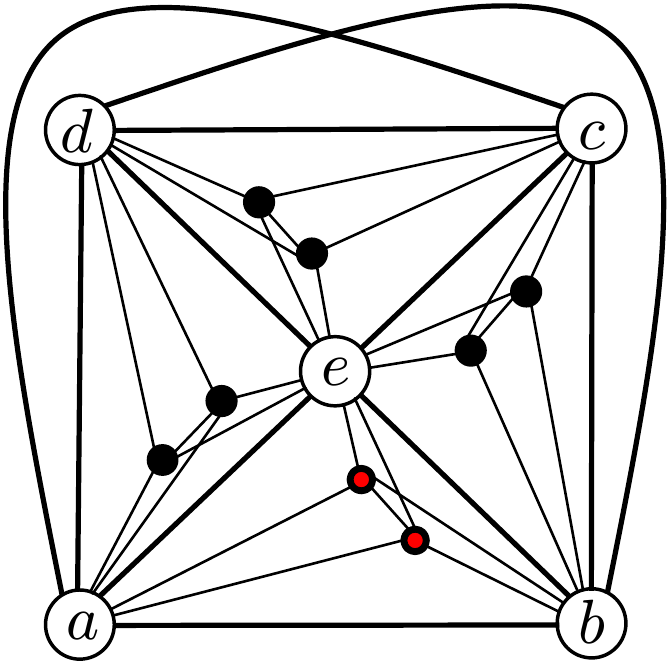}%
      \phantomsubcaption\FigLabel{nm-2}};
    \node[anchor=north east] at (f2.north west) {(\subref*{fig:nm-2})};
  \end{tikzpicture}
  \caption{The embedding in (\subref{fig:nm-1}) is maximal \NIC{}-planar,
     however, re-embedding the $K_5$ subgraph with the vertices $a, b, e$
     admits the addition of the edge $\Edge{b}{d}$ in the outer face, which
     yields the maximal \NIC{}-planar graph in (\subref{fig:nm-2}).}%
  \label{fig:nic-maximality}
\end{figure}
\begin{conjecture}\ConjLabel{characterization}
If $G$ is a graph with a maximal \NIC{}-planar embedding $\Embedding{G}$
that complies with \Lem{dual-adjacencies} and every sequence of
kite flips of a pair of adjacent \TriangleFace{}-nodes that are
either \Free{} or adjacent to a single, common \KiteFace{}-node yields
in turn a maximal \NIC{}-planar embedding, then $G$ is maximal \NIC{}-planar.
\end{conjecture}
Even though the generalized dual has a unique planar embedding, the previous
results already demonstrated that this does not equally apply to maximal
\NIC{}-planar graphs in general.
To obtain a generalized dual, an embedding is needed, which in particular also
determines the pairs of crossing edges.
In fact, even a maximal \NIC{}-planar graph may admit many embeddings.
To begin with, consider the complete graph $K_5$. It has one \One{}-planar
embedding up to isomorphism~\cite{Kyncl-09} and admits three
\One{}-planar and even \IC{}-planar embeddings if the outer face is
fixed, see \Fig{K5fixedface}, which each uses one outer edge in a kite.
\begin{figure}[tb]
   \centering
   \includegraphics[scale=0.27]{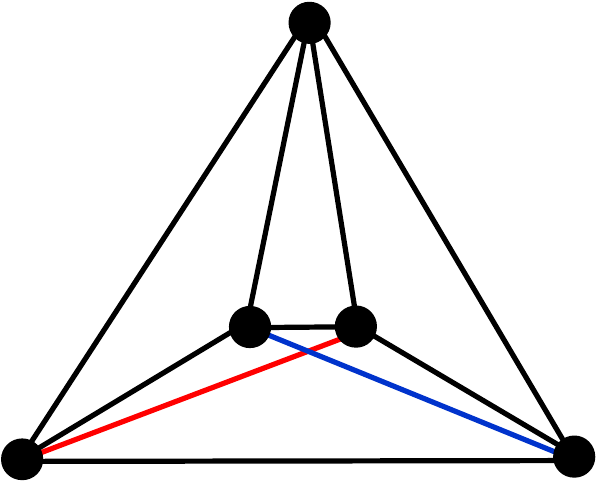}
   \quad
   \includegraphics[scale=0.27]{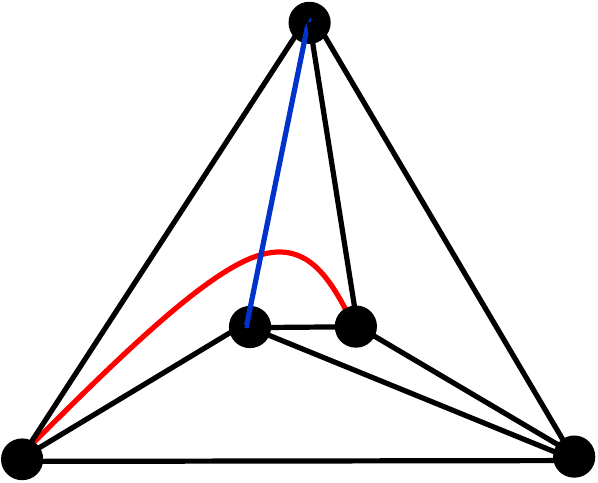}
   \quad
   \includegraphics[scale=0.27]{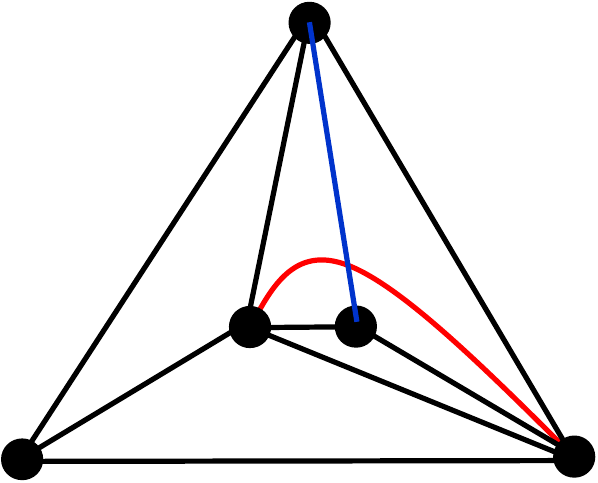}
   \caption{Three embeddings of $K_5$ with a fixed outer face. Each kite includes the edge
   between the inner vertices and one of the outer edges.}%
   \label{fig:K5fixedface}
\end{figure}
Next, we want to study common \One{}-planar embeddings of two $K_5$ graphs.
In general, two subgraphs $H$ and $H'$ are said to be \emph{$k$-sharing}
if they have at least $k$ common vertices.
They \emph{share a crossing}
in an embedding $\Embedding{G}$ of their common supergraph $G$
if there are edges $e$ of $H$ and $e'$ of $H'$ that cross
in $\Embedding{G}$.
\begin{figure}
\centering
\begin{tikzpicture}[node distance=.3cm]
\matrix[matrix of nodes, column sep=.2cm,row sep=0.3cm,nodes={anchor=center}] (m) {%
\includegraphics[scale=.55]{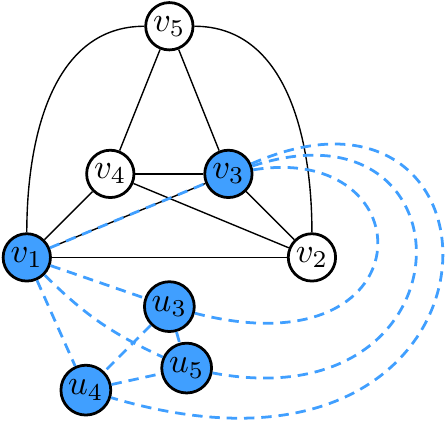}
\phantomsubcaption\label{fig:double-k5-xing-edge}
&
\includegraphics[scale=.55]{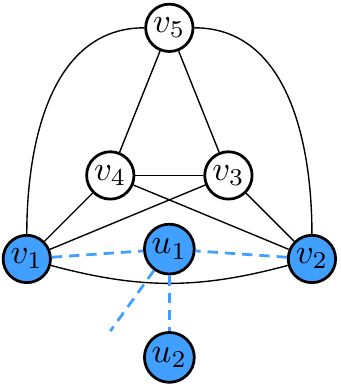}
\phantomsubcaption\label{fig:double-k5-kite-edge}
&
\includegraphics[scale=.55]{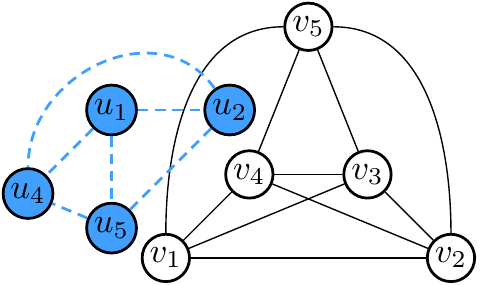}
\phantomsubcaption\label{fig:double-k5-plane-edge}
&
\includegraphics[scale=.55]{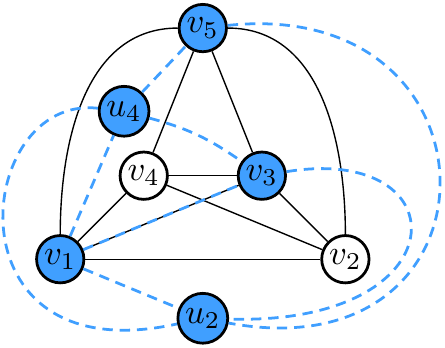}
\phantomsubcaption\label{fig:double-k5-3sharing}
\\
};
\node[fit=(m-1-1) (m-1-2) (m-1-3) (m-1-4)] (row1) {};

\path (m-1-1.north west) |- node[anchor=north west] {(\subref*{fig:double-k5-xing-edge})} (row1.north east);
\path (m-1-2.north west) |- node[anchor=north west] {(\subref*{fig:double-k5-kite-edge})} (row1.north east);
\path (m-1-3.north west) |- node[anchor=north west] {(\subref*{fig:double-k5-plane-edge})} (row1.north east);
\path (m-1-4.north west) |- node[anchor=north west] {(\subref*{fig:double-k5-3sharing})} (row1.north east);
\end{tikzpicture}
\caption{%
Two $K_5$s $\Kfive$ and $\Kfive'$ with a common edge that is crossed by
another edge of $\Kfive$~(\subref{fig:double-k5-xing-edge}),
an edge of $\Kfive'$ crosses a planar kite edge of $\Kfive$~(\subref{fig:double-k5-kite-edge}),
an edge of $\Kfive'$ crosses a planar non-kite edge of $\Kfive$~(\subref{fig:double-k5-plane-edge}),
and
two $3$-sharing $K_5$ subgraphs with a common \One{}-planar embedding~(\subref{fig:double-k5-3sharing}).
Black vertices and solid edges are those of $\Kfive$, colored vertices
and dashed edges (also) belong to $\Kfive'$.}%
\label{fig:double-k5}
\end{figure}
\begin{lemma}\label{lem:3sharing}
If two $K_5$ subgraphs $\Kfive$ and $\Kfive'$ of $G$ share a crossing
in a \One{}-planar embedding $\Embedding{G}$, then they are 3-sharing
and this bound is tight.
\end{lemma}
\begin{proof}
Let $\Kfive = G[v_1,v_2,v_3,v_4,v_5]$
and $\Kfive' = G[u_1,u_2,u_3,u_4,u_5]$
and consider the embeddings $\Embedding{\Kfive}$ and $\Embedding{\Kfive'}$
inherited from $\Embedding{G}$.
Both $\Embedding{\Kfive}$ and $\Embedding{\Kfive'}$ are unique up
to isomorphism, as depicted in \Fig{K5fixedface}.
\Wilog, assume that $\Edge{v_1}{v_3}$ crosses $\Edge{v_2}{v_4}$ in
$\Embedding{\Kfive}$, \ie, $\Kfour = G[v_1,v_2,v_3,v_4]$ is embedded as a kite.

Suppose that one of $\Edge{v_1}{v_3}$ and $\Edge{v_2}{v_4}$
is also an edge of $\Kfive'$ (see~\Fig{double-k5-xing-edge}).
This in particular includes the cases that the shared crossing involves one of
$\Edge{v_1}{v_3}$ or $\Edge{v_2}{v_4}$ and that the edge of $\Kfive'$ that
crosses an edge of $\Kfive$ is part of $\Kfive$.
\Wilog, assume that $u_1 = v_1$ and $u_2 = v_3$.
If $\Kfive$ and $\Kfive'$ are at most $2$-sharing, there must be three
further vertices $u_3, u_4, u_5$ of $\Kfive'$ that are not part of $\Kfive$.
Note that $u_3, u_4, u_5$ form a $K_3$ $\Triangle$.
As neither of $\Edge{v_1}{v_3}$ and $\Edge{v_2}{v_4}$ may be crossed
a second time and $\Edge{v_1}{v_2}$, $\Edge{v_2}{v_3}$, $\Edge{v_3}{v_4}$,
$\Edge{v_4}{v_1}$ may be crossed at most once, none of
$\Triangle$'s vertices can be incident to a face of $\Kfour$.
Hence, they must be incident to one or more faces that are also incident to
$v_5$, \ie, one of the \Trivial{} triangle faces of $\Embedding{\Kfive}$.
Due to \One{}-planarity and $v_5$ having vertex degree four, $\Triangle$'s
vertices must all reside in the same face.
Then, however, either one of $\Triangle$'s is crossed twice by
edges incident to $v_5$ or vice versa, a contradiction to
the \One{}-planarity of $\Embedding{G}$.

Suppose that a planar edge of $\Embedding{\Kfour}$, \wilog, $\Edge{v_1}{v_2}$,
is crossed by an edge $\Edge{u_1}{u_2}$ of $\Kfive'$, where, again \wilog,
$u_1$ is the vertex in one of the non-\Trivial{} triangle faces
of $\Kfour$ (see~\Fig{double-k5-kite-edge}).
Then, besides $u_2$, $u_1$ can be adjacent to at most $v_1$ and $v_2$,
but not to any fifth vertex of $\Kfive'$, irrespective of whether this
vertex is also in $\Kfive$ or not.

Finally, suppose that one of the edges incident to $v_5$, \wilog, $\Edge{v_1}{v_5}$
is crossed by an edge $\Edge{u_1}{u_2}$ of $\Kfive'$.
This situation is depicted in \Fig{double-k5-plane-edge}.
Note that the case where $\Edge{u_1}{u_2}$ is also an edge of $\Kfive$ has
already been considered above.
Hence, we can assume \wilog, that $u_1$ is not a vertex of $\Kfive$.
If $\Kfive$ and $\Kfive'$ are at most $2$-sharing, there must be again
at least two further vertices $u_4$ and $u_5$ that are adjacent to
each other as well as to $u_1$ and $u_2$.
In particular, every triple of $u_1$, $u_2$, $u_4$, and $u_5$ forms a $K_3$.
By the above argument, if the vertices of such a $K_3$ are incident to a
\Trivial{} triangle face of $\Embedding{\Kfive}$, they must all reside in the
same face.
As $\Edge{u_1}{u_2}$ crosses $\Edge{v_1}{v_5}$, this does not apply
to the two $K_3$s containing both $u_1$ and $u_2$, a contradiction.
Note that our argumentation does not exclude the possibility that
$u_2 = v_4$ or that $v_1$ and/or $v_5$ are also vertices of $\Kfive'$.
However, if $\Kfive'$ contains $v_4$, $v_1$, and $v_5$, then $\Kfive$ and
$\Kfive'$ are $3$-sharing.

\Fig{double-k5-3sharing} shows that if $\Kfive$ and $\Kfive'$ are $3$-sharing,
they may share a crossing in $\Embedding{G}$, which testifies that this bound
is tight.
\end{proof}
Note that the converse of \Lem{3sharing} is not true. Any trivial triangle
of an embedded $K_5$ can contain two additional vertices such that these
plus the three triangle vertices form a second $K_5$.
We will now use this result to prove an exponential number of embeddings.
\begin{lemma}
\LemLabel{nestedtriangles}
There are maximal \NIC{}-planar graphs with an exponential number of
\NIC{}-planar embeddings.
\end{lemma}
\begin{proof}
  We construct graphs $G_k$ in two stages. First, consider the nested
  triangle graph $T_k$~\cite{dlt-pepg-84, dett-gdavg-99} with vertices $u_i,
  v_i, w_i$ at the $i$-th layer for $i = 1, \ldots, k$ in counter-clockwise
  order. There are \emph{layer edges} $\Edge{u_i}{v_i}, \Edge{v_i}{w_i},
  \Edge{w_i}{u_i}$ for $i=1,\ldots,k$ and \emph{intra-layer edges}
  $\Edge{u_i}{u_{i-1}}, \Edge{v_i}{v_{i-1}}, \Edge{w_i}{w_{i-1}}$ and
  $\Edge{u_i}{w_{i-1}}$, $\Edge{v_i}{u_{i-1}}$, $\Edge{w_i}{v_{i-1}}$ for $1
  \leq i < k$, which triangulate $T_k$ by planar edges.

  For every triangle $t$ with vertices $u, v, w$ we insert two vertices $a$
  and $b$ and add edges $\Edge{u}{a}, \Edge{u}{b}, \Edge{v}{a}, \Edge{v}{b},
  \Edge{w}{a}, \Edge{w}{b}$ such that these five vertices together induce $K_5$.
	\Fig{one-layer} shows one layer of $G_k$.

  \begin{figure}[tb]
    \centering
    \includegraphics[scale=0.45]{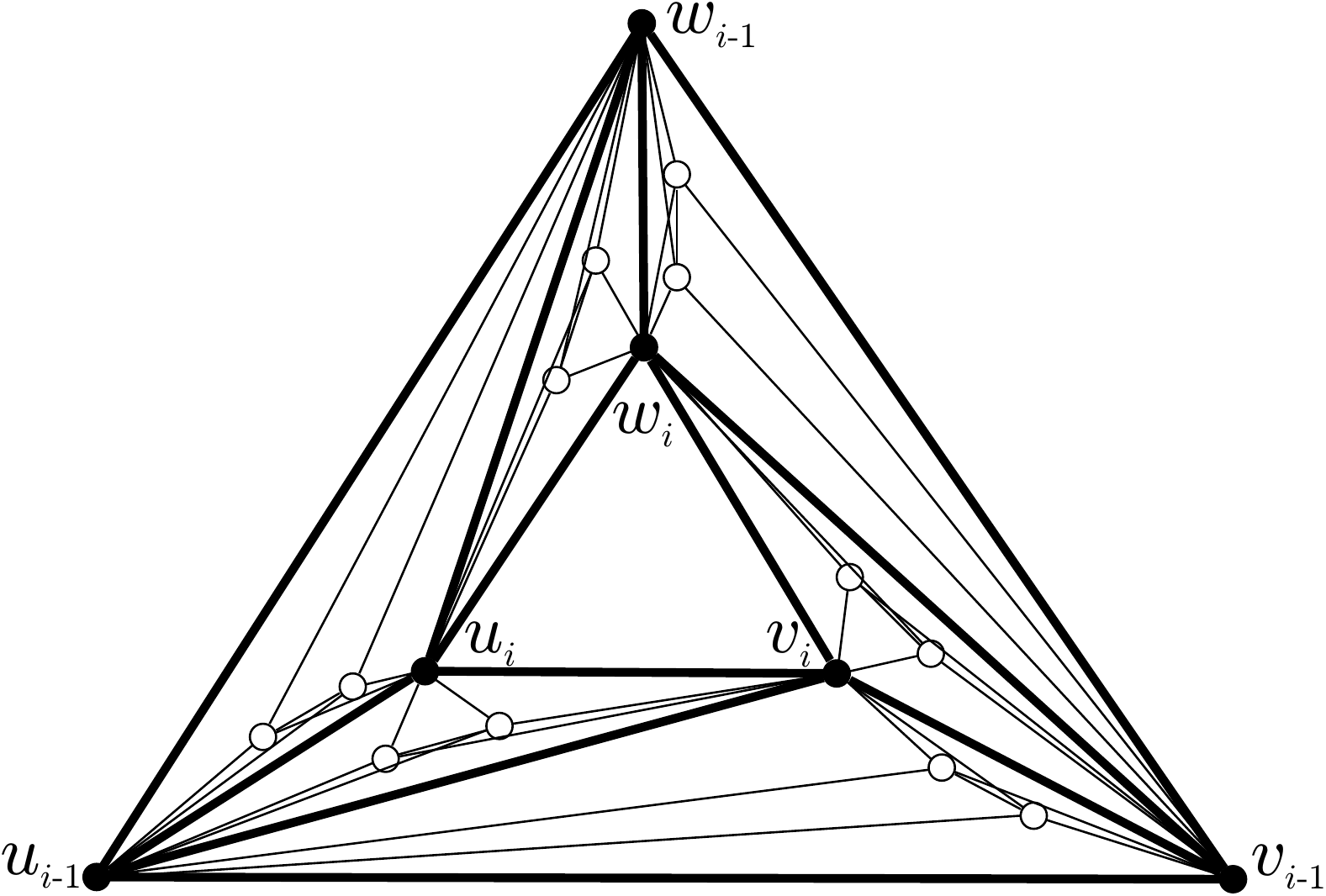}
    \caption{The layer $i$ of $G_k$.}%
    \label{fig:one-layer}
  \end{figure}

	As every pair of $K_5$s in $G_k$ is 2-sharing and every edge is part of at
	least one $K_5$, the edges shared by two $K_5$s are edges of $T_k$ and
	embedded planar in $\Embedding{G_k}$ by \Lem{3sharing}.
  Since $G_k[T_k]$ is triconnected it has a unique planar embedding. Every
  triangle $t$ of $\Embedding{T_k}$ includes two vertices such that there
  is a $K_5$. No further edge can be added without violating maximal
  \NIC{}-planarity. Hence, $G_k$ is maximal \NIC{}-planar.

  For each layer, there are at least two \NIC{}-planar embeddings, one with the
  edges $\Edge{u_i}{u_{i-1}}, \Edge{v_i}{v_{i-1}}, \Edge{w_i}{w_{i-1}}$ in
  kites and one with the edges $\Edge{u_i}{w_{i-1}}$, $\Edge{v_i}{u_{i-1}},
  \Edge{w_i}{v_{i-1}}$. In either case, the layer edges are not part of a
  kite so that the embeddings in the layers are independent.
  We thereby obtain at least $2^{k - 1}$ \NIC{}-planar embeddings. In order to
  distinguish identical embeddings up to a rotational symmetry, we fix the
  embedding of the $K_4$ in the first two layers so that at least $2^{k -
    3}$ different embeddings remain for graphs of size $15k - 12$.
  \Qed
\end{proof}
Planar embeddings have been generalized to maps~\cite{cgp-mg-02} so
that there is an adjacency between faces if their boundaries
intersect.
Also the intersection in a point suffices for an adjacency.
There is a $k$-point $p$ if $k$ faces meet at $p$.
 A hole-free map
$\mathcal{M}$ defines a hole-free graph $G$ so that the faces of
$\mathcal{M}$ correspond to the vertices of $G$ and there is an edge
if and only if the respective regions are adjacent.
(There is a hole if a region is not associated with a vertex).
Obviously,  a $k$-point induces $K_k$ as a subgraph of $G$.
If no more than $k$ regions meet at a point, then $\mathcal{M}$ is a $k$-map
and $G$ is a hole-free $k$-\emph{map graph}.
Chen \ea{}~\cite{cgp-mg-02,cgp-rh4mg-06} observed that a graph is triangulated
\One{}-planar if and only if it is a $3$-connected hole-free $4$-map graph.
For a triangulation it suffices that at least one embedding is triangulated.
For \NIC{}-planar graphs, this implies a characterization with an independent
set.
\begin{corollary}%
\label{cor:triang-vs-4map}
A graph $G$ is triangulated \NIC{}-planar if and only if $G$ is the graph of a
hole-free $4$-map $\mathcal{M}$ such that the $4$-points of $\mathcal{M}$
are an independent set.
\end{corollary}
\begin{proof}
Consider an embedding $\Embedding{G}$ of a triangulated
\NIC{}-planar graph $G$ and remove all pairs of crossing edges.
Then, the planar dual of the obtained graph is a map $\mathcal{M}$
such that every $4$-point corresponds one-to-one to a kite in
$\Embedding{G}$. The statement follows from the observation that two
$4$-points are adjacent in $\mathcal{M}$ if and only if two kites in
$\Embedding{G}$ share an edge.
\Qed
\end{proof}

\section{Density}
\SectLabel{density}
For the analysis of densest and sparsest maximal \NIC{}-planar graphs, we
use the \emph{discharging method} that was successfully applied for the
proof of the 4-color theorem \cite{ah-epgfc-77} and improving upper bounds
on the density of quasi-planar graphs \cite{at-mneqp-07}. In this technique
one assigns ``charges'' to the vertices and faces of a planar graph and
computes the total charge after assigning all charges and after a
redistribution (discharging phase).

We assign charges as follows: Consider the \KiteFace{}-nodes and their neighborhood in $\Dual{G}$.
Let $\Embedding{G}$ be an arbitrary \NIC{}-planar embedding of a
maximal \NIC{}-planar graph $G$.
Define the \emph{level} $\Level(\DualNode)$ of a node $\DualNode$ of $\Dual{G}$
as its minimum distance to a \KiteFace{}-node.
Thus, every \KiteFace{}-node has level $0$, and every
\TriangleFace{}- or \TetrahedronFace{}-node that is adjacent to a
\KiteFace{}-node has level $1$.
A node with level $x$ is also called an $\Level_x$-node.
\Lem{dual-adjacencies} immediately implies:
\begin{corollary}
\CorLabel{levels}
Let $\Dual{G}$ be the generalized dual
with respect to a \NIC{}-planar embedding $\Embedding{G}$
of a maximal \NIC{}-planar graph $G$ with $n \geq 5$.
Then, the level of a node is at most $2$,
every \TetrahedronFace{}-node has level $1$ and
can be adjacent only to $\Level_0$- and $\Level_1$-nodes,
and every $\Level_2$-\TriangleFace{}-node %
is adjacent to at most one other $\Level_2$-\TriangleFace{}-node. %
\end{corollary}
Consider a \KiteFace{}-node $\DualNode$.
Call the set of all nodes whose level equals their distance to $\DualNode$
the \emph{sphere} $\Sphere(\DualNode)$ of $\DualNode$.
Note that the spheres of two \KiteFace{}-nodes need not be disjoint
(see, \eg, node $t_1$ in \Fig{kite-spheres})
and recall that a \KiteFace{}-node has four adjacencies.
If we only consider a single adjacency
and divide the \TriangleFace{}- and \TetrahedronFace{}-nodes of a sphere
equally among all spheres they belong to,
we obtain pairwise disjoint \emph{quarter spheres}.
These quarter spheres are the elements of discharging.
In consequence, we obtain fractions of nodes, \eg,
\TriangleFaceFrac{}{2}- or
\TetrahedronFaceFrac{}{3}-nodes.
\begin{figure}[tb]
\centering
\includegraphics[width=.55\textwidth]{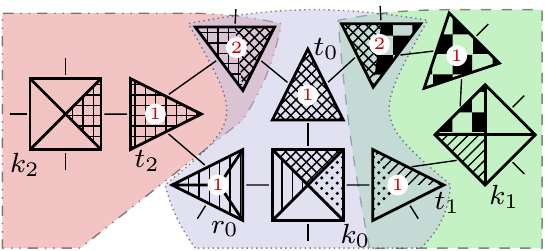}
\caption{%
Clipping of a generalized dual graph, showing parts of the spheres (shaded)
and quarter spheres (fill patterns) of three \protect\KiteFace{}-nodes.
The numbers correspond to the nodes' levels.}
\label{fig:kite-spheres}
\end{figure}
\Fig{kite-spheres} depicts a clipping of the generalized dual of a graph.
It shows three \KiteFace{}-nodes, whose spheres are indicated by the respective
shaded regions.
For the left, center, and right sphere, the patterns highlight one, three, and
two quarter spheres, respectively.
As no two \KiteFace{}-nodes can be adjacent, a quarter sphere can never be
empty, but contains at least a fraction of a \TriangleFace{}- or
\TetrahedronFace{}-node.
We denote the quarter sphere of $\DualNode$ with respect to its neighbor
$\DualNodeOther$ by $\QuarterSphere(\DualNode,\DualNodeOther)$.

\begin{lemma}
\LemLabel{quarter-spheres}
Every quarter sphere %
of a \KiteFace{}-node $\DualNode$
in the generalized dual of a maximal \NIC{}-planar graph with respect to
a \NIC{}-planar embedding consists of at least
either a \TriangleFaceFrac{}{3}- or
a \TetrahedronFaceFrac{}{3}-node.
It can at most be attributed either
two \TriangleFace{}-nodes,
or a \TriangleFace{}-, a \TetrahedronFaceFrac{}{4}-, and a \TriangleFaceFrac{}{4}-node,
or a \TetrahedronFaceFrac{}{2}- and a \TriangleFaceFrac{}{2}-node.
\end{lemma}
\begin{proof}
Observe that by definition, $\Sphere(\DualNode)$ cannot contain
another \KiteFace{}-node besides $\DualNode$ and thus, neither can
its quarter spheres.
Consider a quarter sphere $\QuarterSphere(\DualNode,\DualNodeOther)$ of a
\KiteFace{}-node $\DualNode$ that is adjacent to a node $\DualNodeOther$.
Then, $\DualNodeOther$ is a \TetrahedronFace{}-node or a
\TriangleFace{}-node.

We assume first that $\DualNodeOther$ is a \TetrahedronFace{}-node.
Let $\DualNodeThird \neq \DualNode$ and $\DualNodeFourth \neq \DualNode$ denote
the two other neighbors of $\DualNodeOther$ besides $\DualNode$.
By \Cor{levels}, $\Level(\DualNodeThird) \leq 1$ and
$\Level(\DualNodeFourth) \leq 1$.
Thus, no part of them can be contained in
$\QuarterSphere(\DualNode,\DualNodeOther)$, as their distance via
$\DualNodeOther$ to a \KiteFace{}-node would be $2$.
In consequence of \Lem{dual-adjacencies}, neither $\DualNodeThird$ nor
$\DualNodeFourth$ is a \TetrahedronFace{}-node.
If both are \KiteFace{}-nodes, $\DualNodeOther$ is shared
among three spheres, if only one of them is a \KiteFace{}-node,
it is shared among two spheres, and if both are \TriangleFace{}-nodes,
it belongs entirely to $\QuarterSphere(\DualNode,\DualNodeOther)$.
Hence, $\Sphere_{\DualNodeOther}(\DualNode)$ contains at least
a \TetrahedronFaceFrac{}{3}-node.

For the upper bound, we assume that $\DualNodeOther$
is contained entirely in $\QuarterSphere(\DualNode,\DualNodeOther)$.
Then, $\DualNodeThird$ and $\DualNodeFourth$ must be \TriangleFace{}-nodes,
which are in turn adjacent to at least one \KiteFace{}-node each.
We split $\QuarterSphere(\DualNode,\DualNodeOther)$ once more into
two \emph{semi-quarter spheres} that each contains one half of $\DualNodeOther$
and, depending on their belonging to $\QuarterSphere(\DualNode,\DualNodeOther)$,
the (fractions of) $\DualNodeThird$ and $\DualNodeFourth$,
respectively.
As $\Dual{G}$ is simple and a node's maximum level is $2$, the semi-quarter
spheres are disjoint.
\Wilog, consider $\DualNodeThird$ and let $\DualNodeThird'$ be
a \KiteFace{}-node adjacent to $\DualNodeThird$.
Then, one of the semi-quarter spheres of $\QuarterSphere(\DualNodeThird',\DualNodeThird)$
consists of exactly one \TriangleFaceFrac{}{2}-node.
Hence, for every semi-quarter sphere containing a \TetrahedronFaceFrac{}{2}-node,
there is at most one other semi-quarter sphere containing a \TriangleFaceFrac{}{2}-node.
We can therefore attribute to each of both semi-quarter spheres a \TetrahedronFaceFrac{}{4}-node
and a \TriangleFaceFrac{}{4}-node.

Assume now that $\DualNodeOther$ is a \TriangleFace{}-node.
Then, due to \Cor{levels}, $\QuarterSphere(\DualNode,\DualNodeOther)$ cannot
contain a \TetrahedronFace{}-node at all.
As above, $\QuarterSphere(\DualNode,\DualNodeOther)$ must contain
$\DualNodeOther$ with a proportion of at least $\frac{1}{3}$.
This lower bound is reached if $\DualNodeOther$ is adjacent to three
\KiteFace{}-nodes and hence shared among three quarter spheres.

We obtain a maximal quarter sphere if $\DualNodeOther$ is entirely part of
$\QuarterSphere(\DualNode,\DualNodeOther)$.
The case that $\DualNodeOther$ is adjacent to a \TetrahedronFace{}-node is
already covered in the above argument from the viewpoint of the
\TetrahedronFace{}-node.
Thus, we only address the case where a \TriangleFace{}-node is adjacent
to $\DualNodeOther$.
Consider again the semi-quarter spheres.
By \Cor{levels}, $\DualNodeOther$ may be adjacent to a $\Level_2$-\TriangleFace{}-node
$\DualNodeThird$.
As a $\Level_2$-\TriangleFace{}-node can be adjacent to at most one other
$\Level_2$-\TriangleFace{}-node, $\DualNodeThird$ must be adjacent to
another $\Level_1$-node besides $\DualNodeOther$, \ie, $\DualNodeThird$ must be
shared between at least two quarter spheres.
Subsequently, the semi-quarter sphere consists of at most two
\TriangleFaceFrac{}{2}-nodes.

As indicated above and depicted in \Fig{kite-spheres}, it is possible for a
$\Level_1$-\TriangleFace{}-node to be adjacent to both a
$\Level_2$-\TriangleFace{}-node and a \TetrahedronFace{}-node.
Hence, we obtain the three stated combinations.
\Qed{}
\end{proof}
The generalized dual graph introduced in the previous section
enables us to estimate the density of maximal \NIC{}-planar graphs.
\begin{lemma}
\LemLabel{maximal-minimal-spheres}
Let $G$ be a maximal \NIC{}-planar graph.
If $G$ is optimal, then every quarter sphere with respect to any \NIC{}-planar
embedding of $G$ consists of a \TriangleFaceFrac{}{3}-node and $G$ has four
planar edges for every pair of crossing edges.
If $G$ is sparsest maximal, then every quarter sphere
can at most be attributed either
two \TriangleFace{}-nodes,
or a \TriangleFace{}-, a \TetrahedronFaceFrac{}{4}-, and a \TriangleFaceFrac{}{4}-node,
or a \TetrahedronFaceFrac{}{2}- and a \TriangleFaceFrac{}{2}-node.
For every pair of crossing edges, $G$ has 14 planar edges.
\end{lemma}
\begin{proof}
Let $\Embedding{G}$ be a \NIC{}-planar embedding of a maximal \NIC{}-planar
graph $G$ and denote by $\Dual{G}$ the corresponding generalized dual graph.
Consider the quarter spheres of the \KiteFace{}-nodes in $\Dual{G}$.
As every sphere contains exactly one \KiteFace{}-node and otherwise only
nodes representing planar edges,
we obtain a lower bound on the density of $G$ if the quarter spheres' sizes
are maximized and an upper bound if they are minimized.
Recall that every edge segment separates two faces.
Thus, we assign to every face only one half of the edges on its boundary.
A \TriangleFace{}-node represents one face which is a \Trivial{} triangle.
Thus, every \TriangleFace{}-node corresponds to $\frac{3}{2}$ planar edges.
A \TetrahedronFace{}-node represents three faces each of which is a \Trivial{}
triangle and therefore corresponds to $\frac{9}{2}$ planar edges.
The set of faces represented by a \KiteFace{}-node consists of four
non-\Trivial{} triangles, which yields a pair of crossing edges and
$\frac{4}{2}$ planar edges.

Due to \Lem{quarter-spheres}, a quarter sphere contains at least
$\frac{1}{3}$ of a \TetrahedronFace{}-node or $\frac{1}{3}$ of
a \TriangleFace{}-node, which corresponds to $\frac{3}{2}$ and $\frac{1}{2}$
planar edges, respectively.
The smaller the ratio of planar edges per crossing edge the
larger the density of $G$.
Thus, we obtain a minimum of $4\cdot\frac{1}{2} + 2 = 4$ planar edges and one
pair of crossing edges for an entire sphere, including the edges represented by
the \KiteFace{}-node.
On the other hand, a quarter sphere can be attributed at most either two
\TriangleFace{}-nodes or a \TriangleFace{}-node, a
\TetrahedronFaceFrac{}{4}-node, and a \TriangleFaceFrac{}{4}-node, or a
\TetrahedronFaceFrac{}{2}-node and a \TriangleFaceFrac{}{2}-node---which
corresponds to $3$ planar edges in all three cases.
For the entire sphere, this yields a total of $4 \cdot 3 + 2 = 14$ planar edges
and one pair of crossing edges.
\Qed{}
\end{proof}
In case of optimal \NIC{}-planar graphs, \Lem{maximal-minimal-spheres}
implies:
\begin{corollary}
\CorLabel{densest-dual-structure}
The generalized dual graph of every \NIC{}-planar embedding
of an optimal \NIC{}-planar graph is bipartite with vertex
set $\Dual{V}_{\scalebox{.6}{\KiteFace}} \cupdot \Dual{V}_{\scalebox{.6}{\TriangleFace}}$
such that $\Dual{V}_{\scalebox{.6}{\KiteFace}}$
contains only \KiteFace{}-nodes and $\Dual{V}_{\scalebox{.6}{\TriangleFace}}$
contains only \TriangleFace{}-nodes.
\end{corollary}
With respect to \NIC{}-planar embeddings of optimal \NIC{}-planar
graphs, this unambiguity immediately yields:
\begin{corollary}
\CorLabel{optimal-unique-embedding}
The \NIC{}-planar embedding of an optimal \NIC{}-planar graph is unique
up to isomorphism.
\end{corollary}
Note that every optimal 1-planar graph has one, two, or eight 1-planar
embeddings \cite{s-o1pgts-10}, which are again unique up to isomorphism.

\begin{figure}[tb]
\centering
\begin{tikzpicture}[node distance=.2cm]
\node[inner sep=0pt,draw=none] (l) {
  \includegraphics[scale=.6]{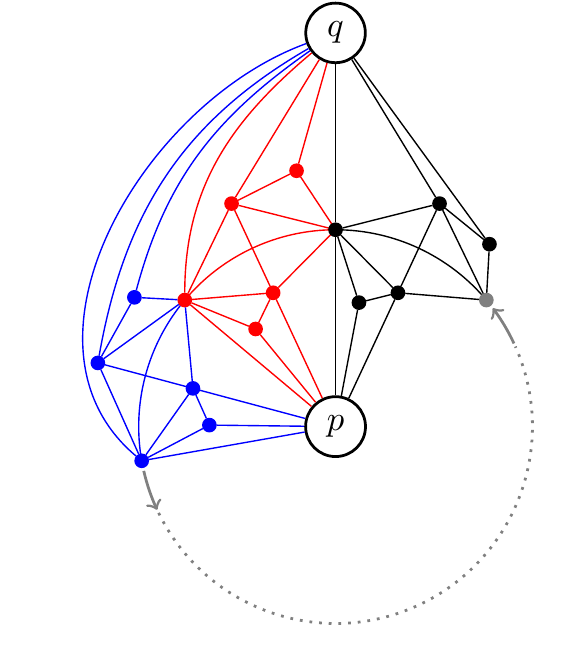}
  \phantomsubcaption\label{fig:sparse-example}};
\node[anchor=north] at (l.north west) {(\subref*{fig:sparse-example})};
\node[inner sep=0pt,draw=none] (m) [right=of l] {
  \includegraphics[scale=.6]{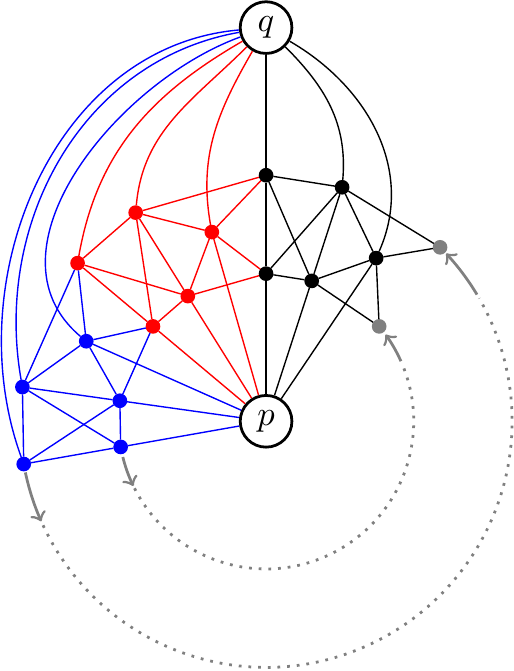}
  \phantomsubcaption\label{fig:dense-example2}};
\path (l.north east) -| node[anchor=north] {(\subref*{fig:dense-example2})} (m.west);
\node[inner sep=0pt,draw=none] (r) [right=of m] {
  \includegraphics[scale=.45]{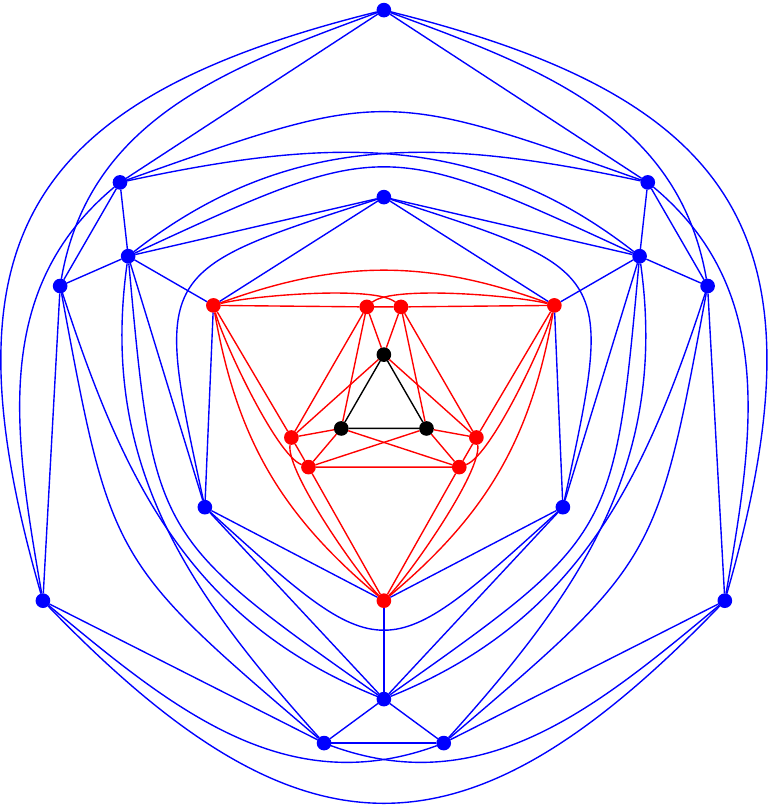}
  \phantomsubcaption\label{fig:dense-example}};
\path (l.north east) -| node[anchor=north] {(\subref*{fig:dense-example})} (r.west);
\end{tikzpicture}
\caption{Construction of sparsest maximal (\subref{fig:sparse-example})
and optimal (\subref{fig:dense-example2},\subref{fig:dense-example})
\NIC{}-planar graphs.}
\label{fig:sparse-dense-example}
\end{figure}
\Lem{maximal-minimal-spheres} suffices to prove upper
and lower bounds on the density of maximal \NIC{}-planar graphs.
Whereas the upper bound of $18/5(n-2)$ edges was already proven by
Zhang~\cite{z-dcmgprc-14} and shown to be tight by Czap and
{\v S}ugerek~\cite{cs-tc1pg-14}, the lower bound has never been assessed
before.
\begin{theorem}
\ThmLabel{density}
Every maximal \NIC{}-planar graph on $n$ vertices with $n \geq 5$
has at least $\frac{16}{5}(n-2)$
and at most $\frac{18}{5}(n-2)$ edges.
Both bounds are tight for infinitely many values of $n$.
\end{theorem}
\begin{proof}
Let $G$ be a maximal \NIC{}-planar graph with $n \geq 5$ vertices
and $m$ edges.
Let $\Embedding{G}$ be a \NIC{}-planar embedding of $G$ and
$\Dual{G}$ be the corresponding generalized dual graph.
Consider the planar subgraph $\PSub{G} \subseteq G$ that is obtained
by removing exactly one of each pair of crossing edges.
Then, $\PSub{G}$ has $n$ vertices and $\PSub{m} \leq m$ edges.
As $\PSub{m}$ is a triangulated planar graph, $\PSub{m} = 3n - 6$.

By \Lem{maximal-minimal-spheres}, there are at least 4 and at most 14 planar
edges per pair of crossing edges.
Furthermore, $\PSub{G}$ contains all planar edges as well as one edge
from each pair of crossing edges.
Hence, $m$ and $\PSub{m}$ differ by at most $\frac{\PSub{m}}{5}$ and
at least $\frac{\PSub{m}}{15}$.
Thus, $\PSub{m} + \frac{\PSub{m}}{15} \leq m \leq \PSub{m} + \frac{\PSub{m}}{5}$,
which yields with $\PSub{m} = 3n - 6$ that
$\frac{16}{5}n - \frac{32}{5} \leq m \leq \frac{18}{5}n - \frac{36}{5}$.

\Fig{sparse-dense-example} shows how to construct a family of
maximal \NIC{}-planar graphs that meet the lower (\Fig{sparse-example})
and upper (\Fig{dense-example2},\subref{fig:dense-example}) bound exactly.
In both cases, the blue subgraph can be copied arbitrarily often and attached
either circularly (\Fig{sparse-example},\subref{fig:dense-example2}) or to the outside
(\Fig{dense-example}).
To obtain a sparsest maximal graph with $n = 7$, we take the red subgraph plus
the leftmost black vertex in \Fig{sparse-example}
and identify $p$ and $q$.
\Qed{}
\end{proof}
\begin{corollary}\CorLabel{construction}
For every $k \geq 1$ there is
a sparsest maximal \NIC{}-planar graph with $n = 5k+2$ and $m = 16k$.
There is an optimal \NIC{}-planar graph with $n = 5k+2$ and $m=18k$
if and only if $k \geq 2$.
\end{corollary}
\begin{proof}
The constructions given in \Fig{sparse-dense-example} show
how to obtain sparsest and densest graphs with $n = 5k+2$.
However, there is no optimal \NIC{}-planar graph with 7 vertices and 18 edges:
Such a graph must have $18 - (3 \cdot 7 - 6) = 3$ pairs of crossing edges.
Consequently, any \NIC{}-planar embedding must have three kites, which is
impossible, as two kites in a \NIC-planar embedding can share at most one
vertex and need seven vertices.
\Qed
\end{proof}
\begin{figure}[tb]
\centering
\begin{tikzpicture}[node distance=.2cm]
\node[inner sep=0pt,draw=none] (l) {
  \includegraphics[scale=.6]{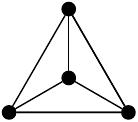}
  \phantomsubcaption\label{fig:maximum-example-n1}};
\node[anchor=north east] at (l.north west) {(\subref*{fig:maximum-example-n1})};
\node[inner sep=0pt,draw=none,anchor=north west] (ll) at ($(l.south west) + (0,-.5cm)$) {
  \includegraphics[scale=.6]{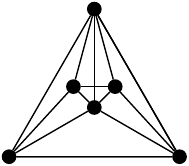}
  \phantomsubcaption\label{fig:maximum-example-n3}};
\node[anchor=north east] at (ll.north west) {(\subref*{fig:maximum-example-n3})};
\node[inner sep=0pt,draw=none,anchor=north west] (m) at ($(l.north east) + (.5cm,0)$) {
  \includegraphics[scale=.6]{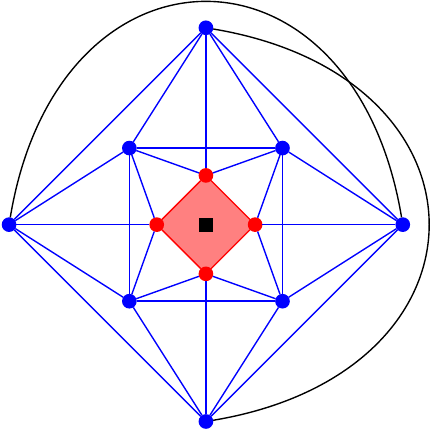}
  \phantomsubcaption\label{fig:maximum-example-n2}};
\path (l.north east) -| node[anchor=north] {(\subref*{fig:maximum-example-n2})} (m.west);
\node[inner sep=0pt,draw=none,anchor=north west] (r) at ($(m.north east) + (.3cm,0)$) {
  \includegraphics[scale=.6]{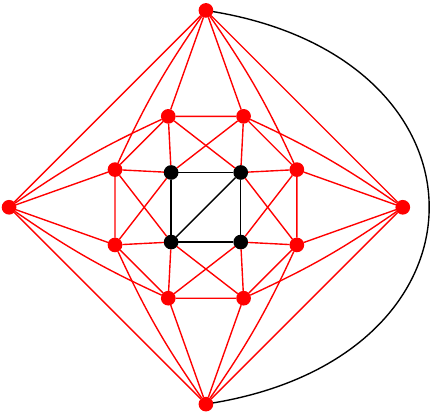}
  \phantomsubcaption\label{fig:maximum-example-n4}};
\path (l.north east) -| node[anchor=north] {(\subref*{fig:maximum-example-n4})} (r.west);
\end{tikzpicture}
\caption{Construction of densest graphs with $n$ vertices such that
$n = 5k + 2 + 1$
(\subref{fig:maximum-example-n1}),
$n = 5k + 2 + 3$
(\subref{fig:maximum-example-n3}),
$n = 5k + 2 + 2$ and $n = 5k + 2 + 4$
(\subref{fig:maximum-example-n2},\subref{fig:maximum-example-n4}).}
\label{fig:nonopt-maximum-examples}
\end{figure}
\begin{theorem}
\ThmLabel{density-intermediate}
For every $i \in \{ 1, \dots, 4 \}$ and infinitely many $k \geq 2$ there is
a \Maximum{} \NIC{}-planar graph with $n = 5k + 2 + i$ and
$m=\lfloor \frac{18}{5}(n-2) \rfloor$.
\end{theorem}
\begin{proof}
For $i = 1$, a \Maximum{} \NIC{}-planar graph $G_1$ with
$m = \lfloor \frac{18}{5}(n-2) \rfloor
= \lfloor \frac{18}{5}(5k + 1) \rfloor
= 18k + 3$ edges can be obtained from any optimal \NIC{}-planar graph $G_0$
by selecting any \Trivial{} triangle in any \NIC{}-planar embedding of $G_0$ and
inserting a vertex $v$ along with three edges, each of which connects $v$ to
one of the three vertices of the triangle
as depicted in \Fig{maximum-example-n1}.
In terms of the generalized dual, this corresponds to replacing one
\TriangleFace{}-node by a \TetrahedronFace{}-node.

For $i = 2$, a \Maximum{} \NIC{}-planar graph $G_2$ with
$m = \lfloor \frac{18}{5}(5k + 2) \rfloor
= 18k + 7$ edges can be constructed
as depicted in \Fig{maximum-example-n2} and \Fig{maximum-example-n4}:
First, we start with the graph shown in \Fig{maximum-example-n4},
which consists of four black and twelve red vertices
and five black and $44$ red edges. %
We, however, do not add the black edge connecting two red vertices.
Instead, we attach eight blue vertices
and $28$ blue edges %
as shown in \Fig{maximum-example-n2}.
This yields a graph with $24 = 5\cdot 4 + 2 + 2$ vertices (\ie, $k = 4$) and
$77$ edges.
The construction can be finished by adding the two black crossing edges
that connect blue vertices as shown again in \Fig{maximum-example-n2}.
$G_2$ thus has $24$ vertices and $79 = 18 \cdot 4 + 7$ edges.
To obtain larger graphs, attach the twelve red vertices and 44 red edges shown
in \Fig{maximum-example-n4} and once more the blue subgraph of
\Fig{maximum-example-n2}. Repeat these two steps arbitrarily often before
finally adding the two black crossing edges.
For each repetition $r$, this yields another $12+8=20$ vertices and $44+28=72$
edges.
Hence, each such graph has $24 + 20r$ vertices (\ie, $k = 4 + 4r$) and $79 +
72r = 18 \cdot (4+4r) + 7$ edges.

For $i = 3$, a \Maximum{} \NIC{}-planar graph $G_3$ with
$m = \lfloor \frac{18}{5}(5k + 3) \rfloor
= 18k + 10$ edges can again be obtained from any optimal \NIC{}-planar graph
$G_0$ by selecting any \Trivial{} triangle in any \NIC{}-planar embedding of $G_0$ and
inserting the subgraph depicted in \Fig{maximum-example-n3}.
More precisely, we add three vertices and ten edges such that the
\TriangleFace{}-node in the corresponding generalized dual is replaced by a
\KiteFace{}-node and four \TriangleFace{}-nodes.

Finally, for $i=4$, a \Maximum{} \NIC{}-planar graph $G_4$ with
$m = \lfloor \frac{18}{5}(5k + 4) \rfloor
= 18k + 14$ edges can be constructed
similarly to the case for $i=2$.
We start again with the graph depicted in \Fig{maximum-example-n4},
which has four black and twelve red vertices
as well as five black and $44$ red edges. %
Here, the construction can already be finished by adding the black
edge connecting two red vertices,
which yields a graph with $4+12 = 16 = 5 \cdot 2 + 2 + 4$ vertices
(\ie, $k = 2$) and $5 + 44 + 1 = 50 = 18 \cdot 2 + 14$ edges.
For a larger graph, we add again the blue subgraph with eight vertices
and $28$ edges as depicted in \Fig{maximum-example-n2}
as well as once more the red subgraph with twelve vertices and $44$ edges
arbitrarily often
prior to connecting the two red vertices by the black edge.
For each repetition $r$, this yields another $8+12=20$ vertices
and $28+44 = 72$ edges.
Hence, each such graph has $16 + 20r$ vertices (\ie, $k = 2 + 4r$) and $50 +
72r = 18 \cdot (2+4r) + 14$ edges.

Note that by \Cor{construction}, optimal graphs, which provide the basis for
the construction in case of $i = 1$ and $i = 3$, only exist for $k \geq 2$ and
that the graphs constructed for $i = 2$ and $i=4$ have $k \geq 4$ and $k \geq
2$, respectively.
\Qed
\end{proof}
From \Cor{construction} and the proof of \Thm{density-intermediate} we obtain:
\begin{corollary}
There are densest \NIC{}-planar graphs for all $n = 5k + 2 + i$ with $i \in
\{0, 1, 3\}$ and all $k \geq 2$.
\end{corollary}
Concerning \IC{}-planar graphs, there are optimal ones
with $\frac{13}{4}n - 6$ edges for all $n = 4k$ and $k \geq 2$
\cite{zl-spgic-13}.
A densest \IC{}-planar graph of size $n \geq 8$ is obtained from an optimal one
of size $4k$ with $k = \lfloor \frac{n}{4} \rfloor$ by the replacement of $i$
\TriangleFace-nodes by \TetrahedronFace{}-nodes, for $i = 1, 2, 3$.
Each \TetrahedronFace{}-node adds one vertex and three edges.
There are sparsest \IC{}-planar graphs with $3n - 5$ edges for all $n
\geq 5$, and this bound is tight \cite{bbh-nipg-17}.

\section{Recognizing Optimal \NIC{}-Planar Graphs in Linear Time}
\SectLabel{densest-recognition}
The study of maximal \NIC{}-planar graphs in \Sect{dual}
also provides a key to a linear-time recognition algorithm for
optimal \NIC{}-planar graphs.
Therefore, we
establish a few more properties of optimal \NIC{}-planar graphs and
their embeddings $\Embedding{G}$.
An edge $e$ is called $k$-fold $K_4$-\emph{covered} in $G$ if $e$ is
part of $k$ $K_4$ subgraphs.
We start with the following observation:
\begin{lemma}
\LemLabel{kite-cover}
Let $\Embedding{G}$ be a \NIC{}-planar embedding of an optimal
\NIC{}-planar graph $G$.
Every edge $e \in E$ is at least 1-fold $K_4$-covered
and there is exactly one $K_4$
in $G$ that contains $e$ and is embedded as a kite in $\Embedding{G}$.
\end{lemma}
\begin{proof}
\Cor{densest-dual-structure} implies that every edge $e \in E$ is at least once
$K_4$-covered.
Moreover, if $e$ is contained in at least two different $K_4$ inducing
subgraphs that are both embedded as a kite, then $e$'s end vertices
are incident to two common pairs of crossing edges, a contradiction
to \NIC{}-planarity.
\Qed
\end{proof}
In \Sect{basics} we already noted that in any \NIC{}-planar
embedding, every subgraph that induces $K_4$ must be embedded
either as a kite or as a tetrahedron, which may in turn be simple
or not.
Recall that in case of a non-simple tetrahedron embedding, $K_4$'s edges may
not cross each other, but they may be crossed by edges that do not belong to
this subgraph.
The following lemma limits the possibilities of how the $K_4$ subgraphs and
their embeddings can interact.
\begin{lemma}
\LemLabel{tetrahedron-uncovered}
Let $\Embedding{G}$ be a \NIC{}-planar embedding of a maximal \NIC{}-planar
graph $G$
and let $\Edge{a}{c}$ and $\Edge{b}{d}$ be two edges that cross each other in
$\Embedding{G}$.
Then one of $\Edge{a}{c}$ and $\Edge{b}{d}$ is 1-fold $K_4$-covered.
\end{lemma}
\begin{figure}[tb]
\centering
\begin{tikzpicture}
\node[inner sep=0pt,draw=none] (l) {
  \includegraphics[scale=.7]{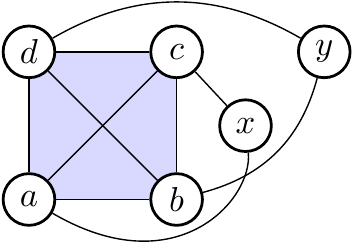}%
  \phantomsubcaption\label{fig:kite-uncover-a}};
\node[anchor=center] at (l.north west) {(\subref*{fig:kite-uncover-a})};
\node[inner sep=0pt,draw=none,anchor=west] (r) at ($(l.east) + (.3cm,0)$) {
  \includegraphics[scale=.7]{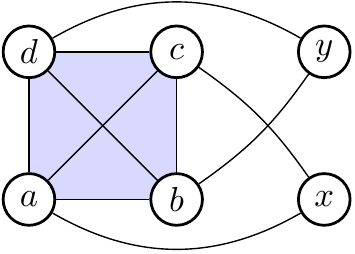}%
  \phantomsubcaption\label{fig:kite-uncover-b}};
\node[anchor=center] at (r.north west) {(\subref*{fig:kite-uncover-b})};
\end{tikzpicture}
\caption{%
Proof of \Lem{tetrahedron-uncovered}.
}
\label{fig:proof-tetra-uncover}
\end{figure}
\begin{proof}
Let $\Kfour = G[a,b,c,d]$.
By \Lem{crossing-implies-kite}, $\Kfour$ is $K_4$.
Suppose for the sake of contradiction that there are subgraphs $\Kfour' \neq
\Kfour$ and $\Kfour'' \neq \Kfour$ of $G$ that both induce $K_4$ and such that
$\Kfour'$ contains $\Edge{a}{c}$ and $\Kfour''$ contains $\Edge{b}{d}$.
Note that $\Kfour' \neq \Kfour''$, otherwise both had vertex set $\{a,b,c,d\}$
and hence, $\Kfour = \Kfour' = \Kfour''$.
Furthermore, both $\Kfour'$ and $\Kfour''$ must be embedded as tetrahedrons,
otherwise, $a,c$ and $b,d$ were incident to two pairs of crossing edges.
Denote by $x \neq y \in V \setminus \{a,b,c,d\}$ two further vertices of $G$
such that $\Kfour'$ contains $x$ and $\Kfour''$ contains $y$.

Consider the closed path $(a,b,c,d)$ of planar edges, which partitions the
set of faces of $\Embedding{G}$ into $P'$ and $P''$.
Due to \Lem{kite-or-triangle}, one of these partitions, \wilog{} $P'$, contains
only the non-\Trivial{} triangles that form the kite embedding of $\Kfour$.
Hence, all faces incident to $x$ and $y$ must reside within $P''$
and the edges or edge segments of $\Edge{a}{x}$, $\Edge{c}{x}$, $\Edge{b}{y}$,
and $\Edge{d}{y}$ only bound faces contained in $P''$.
Consequently, the paths $(a,x,c)$ and $(b,y,d)$ must cross each other in
$P''$.

Recall that $a$, $b$, $c$, and $d$ are already pairwisely incident to a pair of
crossing edges, namely $\Edge{a}{c}$ and $\Edge{b}{d}$.
\Fig{proof-tetra-uncover} shows two of the four possible pairs of additional
crossing edges.
Suppose that $\Edge{a}{x}$ crosses $\Edge{b}{y}$.
Then, $a$ and $c$ are incident to another pair of crossing edges, a
contradiction to the \NIC{}-planarity of $\Embedding{G}$.
Likewise, if $\Edge{c}{x}$ crosses $\Edge{b}{y}$,
or $\Edge{a}{x}$ crosses $\Edge{d}{y}$,
or $\Edge{c}{x}$ crosses $\Edge{d}{y}$,
then $b$ and $c$, or $a$ and $d$, or $c$ and $d$,
respectively, are incident to two common pairs of crossing edges, thereby again
contradicting the \NIC{}-planarity of $\Embedding{G}$.

Subsequently, $\Kfour'$ and $\Kfour''$ cannot both exist.
\Qed
\end{proof}
The combination of \Lem{kite-cover} and \Lem{tetrahedron-uncovered} yields
a characterization of those $K_4$ inducing subgraphs that are embedded as kite:
\begin{corollary}
\CorLabel{densest-iff}
Let $\Kfour$ be a subgraph inducing $K_4$ in an optimal \NIC{}-planar graph $G$
and let $\Embedding{G}$ be a \NIC{}-planar embedding of $G$.
Then, $\Kfour$ is embedded as a kite in $\Embedding{G}$ if and only if
$\Kfour$ has a 1-fold $K_4$-covered edge.
\end{corollary}
\begin{proof}
Let $\{a,b,c,d\}$ denote the vertex set of $\Kfour$.
By \Lem{kite-cover}, every edge of $\Kfour$ is at least once covered by
a $K_4$ which is embedded as a kite.

If $\Kfour$ is embedded as a kite, then one of its crossing edges is 1-fold
$K_4$-covered by \Lem{tetrahedron-uncovered}, and if $\Kfour$ is embedded as a
tetrahedron, then each of its edges is at least $2$-fold $K_4$-covered.
\Qed
\end{proof}
\begin{algorithm}[tb]
\caption{A recognition algorithm for optimal \NIC{}-planar graphs.}
\label{alg:densest-recognition}
\begin{algorithmic}[1]
\Require{graph $G = (V, E)$ with $n = |V| \geq 5$ and $m=|E|$}
\Ensure{\NIC{}-planar embedding $\Embedding{G}$ if $G$ is optimal \NIC{}-planar, else
$\bot$}
\Procedure{TestOptimalNIC}{G}
\If{$m \neq \frac{18}{5}(n-2)$}
\Return $\bot$
\label{algl:num-edges}
\Comment{$G$ is not optimal.}
\EndIf
\State $\KfourSet \gets$ set of $K_4$s in $G$ or $\bot$ in case of timeout
\If{$\KfourSet = \bot$}
\Return $\bot$
\label{algl:arboricity}
\EndIf
\State $\KiteSet \gets \emptyset$
\State create an empty bucket $B[e] = \emptyset$ for each edge $e \in G$
\ForAll{$\Kfour \in \KfourSet$}
\State add $\Kfour$ to every bucket $B[e]$ for every edge $e$ of $\Kfour$
\EndFor
\ForAll{$e \in E$}
\If{$B[e] = \{ \Kfour \}$}
\State $\KiteSet \gets \KiteSet \cup \{\Kfour\}$
\label{algl:lonely-edge}
\Comment{$\Kfour$ must be embedded as kite by \Cor{densest-iff}.}
\EndIf
\EndFor
\ForAll{$e \in E$}
\If{$|B[e] \cap \KiteSet| \neq 1$}
\Return $\bot$
\label{algl:one-kite}
\Comment{\Lem{kite-cover} is violated.}
\EndIf
\EndFor
\State $G' \gets G$
\ForAll{$\Kfour \in \KiteSet$}
\State remove all edges of $\Kfour$ in $G'$
\State add a dummy vertex $z_{\Kfour}$ along with edges to all vertices of $\Kfour$ in $G'$
\EndFor
\If{$G'$ is not planar}
\Return $\bot$
\label{algl:nonplanar-planarization}
\EndIf
\State $\Embedding{G'} \gets$ planar embedding of $G'$
\State $\Embedding{G} \gets$ $\NIC{}$-planar embedding of $G$ obtained from $\Embedding{G'}$
\State \Return $\Embedding{G}$
\EndProcedure
\end{algorithmic}
\end{algorithm}

Now we are ready to prove the main result of this section.
\begin{theorem}
\ThmLabel{densest-recognition}
There is a linear-time algorithm that decides whether a graph is
optimal \NIC{}-planar and, if positive, returns a
\NIC{}-planar embedding.
\end{theorem}
\begin{proof}
Consider the algorithm given in \Alg{densest-recognition}, which
takes a graph $G$ as input and either returns a \NIC{}-planar
embedding $\Embedding{G}$ if $G$ is optimal \NIC{}-planar
and otherwise returns $\bot$.

Let $G = (V,E)$.
First, if the number of edges $m$ of $G$ does not meet the
upper bound of $\frac{18}{5}(n-2)$, $G$ cannot be
optimal \NIC{}-planar.
The algorithm therefore returns $\bot$ in line~\ref{algl:num-edges} if this
check fails.
For the remainder of the algorithm, we can assume that $m \in \bigO(n)$.

Next, we identify those $K_4$ inducing subgraphs of $G$ that must
be embedded as a kite.
To this end, enumerate all subgraphs of $G$ that induce $K_4$ and keep them as
set $\KfourSet$.
This can be accomplished in linear time by running
the algorithm of Chiba and Nishizeki \cite{cn-asl-85} for at most $256 n$
essential steps.
A step is essential if it marks a vertex or an edge.
The inessential steps, like unmark and print, take linear time in the number
of essential steps.
If $G$ is maximal \NIC{}-planar, then $G$ has arboricity four
\cite{n-dfgf-61} and the algorithm completes the computation of the
$K_4$ listing within $256 n$ essential steps.
Otherwise, if the algorithm by Chiba and Nishizeki exceeds the bound on the
running time, $G$ does not have arboricity four and is hence not \NIC{}-planar.
Then our algorithm returns $\bot$ in line~\ref{algl:arboricity}.
Chen \ea \cite{cgp-rh4mg-06} have shown that triangulated 1-planar graphs of
size $n$ have at most $27 n$ $K_4$ subgraphs.
The subset $\KfourSet_{\times}$ of $\KfourSet$ that will later contain
those with kite embeddings is set to $\emptyset$.
Initialize an empty bucket $B[e]$ for every edge $e \in E$.
We now employ a variant of bucket sort on $\KfourSet$ and place a copy
of every element $\Kfour \in \KfourSet$ in all six buckets that represent an
edge of $\Kfour$.
As the size of $\KfourSet$ is linear in the size of $G$, this takes $\bigO(n)$
time.

Afterwards, we loop over the edges of $G$ and apply \Cor{densest-iff}:
If an edge is contained in exactly one $\Kfour \in \KfourSet$, then $\Kfour$
must be embedded as a kite and is therefore added to $\KfourSet_{\times}$ in
line~\ref{algl:lonely-edge}.
This can be accomplished in time $\bigO(m) = \bigO(n)$.
Having identified the $K_4$ inducing subgraphs that must be embedded with a
crossing, we can check whether every edge is contained in exactly one
kite as required by \Lem{kite-cover} again in $\bigO(n)$ time.

The last step in the algorithm consists in identifying the pairs of crossing
edges and, if possible, obtaining a \NIC{}-planar embedding of $G$.
For this purpose, we construct a graph $G'$ from $G$ as follows:
For every $\Kfour \in \KfourSet_{\times}$ with vertex set $\{a,b,c,d\}$,
remove all edges connecting $a$, $b$, $c$, and $d$.
Then, add a new dummy vertex $z_{\Kfour}$ along with edges
$\Edge{a}{z_{\Kfour}}$, $\Edge{b}{z_{\Kfour}}$,
$\Edge{c}{z_{\Kfour}}$, $\Edge{d}{z_{\Kfour}}$.
Thus, $\Kfour$ is replaced by a star with center $z_{\Kfour}$.
This construction requires again $\bigO(n)$ time.

Observe that $G'$ is a subgraph of every planarization of
$G$ with respect to any \NIC{}-planar embedding of $G$.
Hence, if $G'$ is not planar, then $G$ cannot be optimal
\NIC{}-planar, so the algorithm returns $\bot$ in
line~\ref{algl:nonplanar-planarization}.
Otherwise, we obtain a planar embedding $\Embedding{G'}$ of $G$.
This can be done in time $\bigO(n)$.

Next, construct an embedding $\Embedding{G}$ of $G$ from $\Embedding{G'}$
by replacing every dummy vertex $z_{\Kfour}$ and its incident edges
by the edges of $\Kfour$.
The edges incident to $z_{\Kfour}$ are taken as non-\Trivial{} edge segments
and the remaining edges are routed close to these edges.
This corresponds to replacing the four \Trivial{} triangles incident to
$z_{\Kfour}$ in $\Embedding{G'}$ one-to-one by four non-\Trivial{} triangles
forming the kite embedding of $\Kfour$ and takes again $\bigO(n)$ time.
$\Embedding{G}$ now is a \NIC{}-planar embedding of $G$ such that exactly the
elements of $\KfourSet_{\times}$ are embedded with a crossing.
As the number of edges in $G$ meets the upper bound of $3.6(n-2)$
exactly, $G$ is optimal \NIC{}-planar
and the algorithm returns $\Embedding{G}$ as a witness.
The overall running time is in $\bigO(n)$.
\Qed
\end{proof}

\section{Drawing \NIC{}-Planar Graphs}
\SectLabel{nic-is-not-rac}
Every \NIC{}-planar graph is a subgraph of a maximal \NIC{}-planar graph and
every \NIC{}-planar embedding with $n \geq 5$ has a \Trivial{} triangle which
we use as outer face.
By \Cor{triconnected} maximal \NIC{}-planar graphs are triconnected.
Hence, the linear-time algorithm of Alam \ea{}~\cite{abk-sld3c-13}
for \One{}-planar graphs can be used.
\begin{corollary}
  Every \NIC{}-planar graph has a \NIC-planar straight-line drawing on
  an integer grid of $\bigO(n^2)$ size.
\end{corollary}
A graph $G$ has \emph{geometric thickness} $k$ if $G$ admits a straight-line
drawing in the plane such that there is a $k$-coloring of the edges and
edges with the same color do not cross \cite{deh-dtcg-00}.
\begin{corollary}
  Every \NIC{}-planar graph has geometric thickness two.
\end{corollary}

However, \NIC{}-planar graphs do not necessarily admit straight-line
drawings with right angle crossings. In consequence, the classes of
\NIC{}-planar graphs and \RAC{} graphs are incomparable, since \RAC{} graphs
may be too dense \cite{del-dgrac-11}.
\begin{theorem}\ThmLabel{notrac}
There are infinitely many \NIC{}-planar graphs that are not \RAC{} graphs, and
conversely.
\end{theorem}
For the harder part, we construct a \NIC{}-planar graph that is not \RAC{}.
Infinitely many graphs are obtained by multiple copies.
Let $G^+ = (V^+,E^+)$ be obtained from graph $G = (V,E)$ in \Fig{RAC} by
augmenting every planarly drawn edge between two vertices $u$ and $v$ with
seven vertex-disjoint $2$-paths, as shown in \Fig{7-2-paths}.
Every edge of $G$ that is augmented is called a \emph{fat edge}.
If $u$ and $v$ are connected by a fat edge, then $v$ is a \emph{fat neighbor}
of $u$.
We show that graph $G^+$ does not admit a \RAC{} drawing.

Observe that $G$ consists of six $K_4$ and eight $K_3$
subgraphs such that a $K_4$ is attached to each side of a $K_3$.
As shown in \Fig{RAC}, $G$ is \NIC{}-planar and likewise is $G^+$, since the
$2$-paths of each fat edge can be embedded planarly.
We obtain the \emph{induced} \RAC{} drawing $\Drawing{G}$ of $G$ from a \RAC{}
drawing $\Drawing{G^+}$ of $G^+$ by removing the $2$-paths of each fat edge.

Graph $G$ is 4-connected and $G$ and $G^+$ remain biconnected if all pairs
of crossing edges are removed and additionally either a single $K_3$ or a
single $K_4$. After the removal, each vertex of $G$ and $G^+$ still has two
fat neighbors and the remaining graph is connected using only fat edges.
\begin{figure}[tb]
  \centering
  \begin{tikzpicture}
    \node[inner sep=0pt, draw=none] (f1) {
      \includegraphics[scale=1.1]{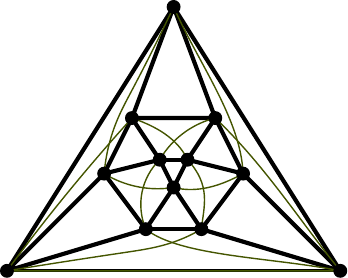}%
      \phantomsubcaption\FigLabel{RAC}};
    \node[anchor=center] at (f1.north west) {(\subref*{fig:RAC})};
    \node[inner sep=0pt, draw=none, anchor=north west] (f2)
        at ($(f1.north east) + (0.3cm,0)$){
      \includegraphics[scale=0.55]{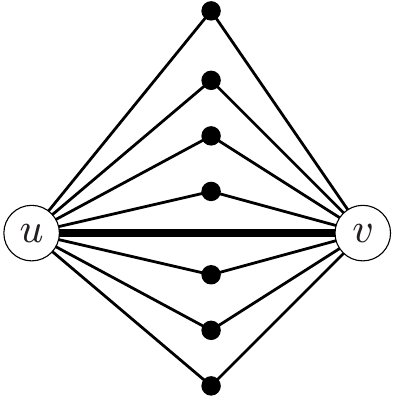}%
      \phantomsubcaption\FigLabel{7-2-paths}};
    \node[anchor=center] at (f2.north west) {(\subref*{fig:7-2-paths})};
    \node[inner sep=0pt, draw=none, anchor=north west] (f3)
        at ($(f2.north east) + (0.5cm,0)$){
      \includegraphics[scale=0.6]{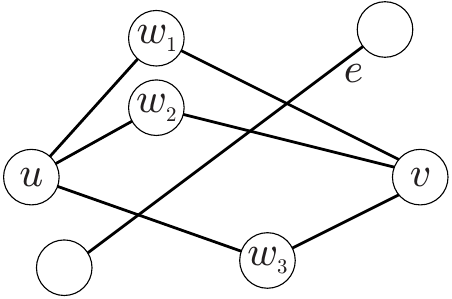}%
      \phantomsubcaption\FigLabel{2-path}};
    \node[anchor=center] at (f3.north west) {(\subref*{fig:2-path})};
  \end{tikzpicture}
  \caption{(\subref{fig:RAC}): The graph $G$. The supergraph $G^+$ extends
        $G$ by seven $2$-paths for each fat edge. $G$ corresponds
        to the red and black part of \Fig{dense-example}.
      (\subref{fig:7-2-paths}): Seven $2$-paths augment an edge to a fat
        edge.
      (\subref{fig:2-path}): A forced fan crossing.}
  \FigLabel{construction}
\end{figure}
The following properties %
are immediate (see \Fig{2-path}):
\begin{lemma}[\!\!\cite{del-dgrac-11}]\LemLabel{fan-crossing}
  A \RAC{} drawing does not admit a \emph{fan-crossing}, \ie, no edge may
  cross two edges with a common end vertex.
\end{lemma}
\begin{lemma}\LemLabel{2-paths}
  If an edge is crossed by $k$ $2$-paths $p_i = (u, w_i, v)$ for $i = 1,
  \ldots, k$ connecting two vertices $u$ and $v$ in a \RAC{} drawing, then $k
  \leq 2$.
\end{lemma}
\begin{lemma}\LemLabel{pigeon}
If there is a triangle $\Triangle$ in a \RAC{} drawing and a fat edge
$\Edge{u}{v}$ so that $u, v$ are not vertices of $\Triangle$, then $u$ and
$v$ are either both inside or both outside $\Triangle$.
\end{lemma}
\begin{proof}
Every vertex $w_i$ of a $2$-path $(u,w_i,v)$ has degree $2$ and thus cannot be
a vertex of $\Triangle$.
If, \wilog, $u$ is inside and $v$ is outside of $\Triangle$, at least one edge
of $\Triangle$ is crossed by at least three $2$-paths, which contradicts
\Lem{2-paths}. \Qed
\end{proof}
\begin{lemma}\LemLabel{K4}
  Let $\Drawing{G^+}$ be a \RAC{} drawing of $G^+$. Then every $K_4$ is
  drawn with a pair of crossing edges.
\end{lemma}
\begin{proof}
  Suppose that
  $G^+[\{u,v,w,x\}]$ is a $K_4$ which
  is not drawn with a pair of crossing edges.
  Then it is drawn as a tetrahedron \cite{Kyncl-09}. 
  \Wilog, let %
  $x$ be inside the triangle
  $f_{uvw}$.
  Then,
  $f_{uvw}$
  is partitioned into three triangles
  $f_{uvx}$, $f_{uwx}$, and $f_{vwx}$.

  Vertex $x$ has a fat neighbor $y \not\in \{v,w,x\}$. %
  By \Lem{pigeon}, $y$ must be inside $f_{uvw}$.
  \Wilog{}, let $y$ be in $f_{uvx}$.
  By the same reasoning, the fat neighbors of $y$ not in $\{u, v, x\}$ must be in
  the same triangle as $y$, which due to the connectivity of $G^+$ and $G$ on
  fat edges even without the $K_3$ $G^+[\{u,v,x\}]$ implies that all vertices
  of $G$ must be in $f_{uvx}$.
  However, $w$ is outside $f_{uvx}$, a contradiction.
  \Qed
\end{proof}
\begin{lemma}\LemLabel{fat-crossings}
  No edge $e \in E$ can cross a fat edge $f \in E$
  in any \RAC{} drawing $\Drawing{G^+}$ of $G^+$.
\end{lemma}
\begin{proof}
  Suppose edge $e = \Edge{u}{v} \in E$ crosses the fat edge $f = \Edge{x}{y}
  \in E$ in $\Drawing{G^+}$. As a fat edge, $f$ shares an edge with a $K_3$
  $\Kthree = G^+[\{x, y, z\}]$, where $z \in V^+ \setminus V$ is a vertex of
  a $2$-path associated with $f$. Edge $e$ cannot cross both $\Edge{x}{y}$ and
  $\Edge{x}{z}$ or both $\Edge{x}{y}$ and $\Edge{y}{z}$ due to forbidden
  fan-crossings by \Lem{fan-crossing}. Thus, \wilog, $u$
  is inside the triangle $f_{xyz}$ and $v$ outside. As there is a
  path $p$ between $u$ and $v$ which consists only of fat edges and is
  vertex disjoint with $\Kthree$,
  there is an edge of $p$ crossing an edge of triangle $f_{xyz}$, a
  contradiction to \Lem{pigeon}. \Qed
\end{proof}

\begin{lemma}\LemLabel{convex}
  Let $\Drawing{G^+}$ be a \RAC{} drawing of $G^+$ and let $\Drawing{G}$ be
  the induced drawing. Then every $K_4$ is drawn as a kite with planar fat
  edges in $\Drawing{G}$.
\end{lemma}
\begin{proof}
  By \Lem{K4}, every $K_4$ $\Kfour$ is drawn with a right angle crossing.
  Let $\Drawing{\Kfour}$ be the induced drawing of $\Kfour$.
  Due to \Lem{fan-crossing}, no other edge of $G$ can
  cross two (or more) edges of $\Kfour$. Furthermore, no other vertex of $G$
  can be inside $\Drawing{\Kfour}$ by the same argument as in the proof of
  \Lem{K4}: The pair of crossing edges partitions $\Drawing{\Kfour}$ into
  four triangles. If a vertex is in such a triangle, then all vertices of
  $G$ must be in the same triangle, since they are connected by fat edges.
  However, the remaining two vertices of $\Kfour$ cannot be in the
  triangle. Hence, $\Drawing{\Kfour}$ is empty and every $K_4$ subgraph of
  $G^+$ is drawn as kite. By \Lem{fat-crossings}, the crossing edges of the
  kites cannot be fat. \Qed
\end{proof}
As a consequence, in the induced drawing $\Drawing{G}$ every $K_3$ in $G$ which is
no subgraph of a $K_4$ is drawn as a \Trivial{} triangle and every $K_4$ in $G$
is drawn as kite, \ie, with no further vertex inside. Finally, consider the outer
face:
\begin{lemma}\LemLabel{outer}
  Let $\Drawing{G^+}$ be a \RAC{} drawing of $G^+$ and let $\Drawing{G}$ be
  the induced drawing. Then the outer face of $\Drawing{G}$ is a \Trivial{} triangle.
\end{lemma}
\begin{proof}
If the outer face of $\Drawing{G}$ is not a \Trivial{} triangle, it
must be one of the four triangles of a $K_4$'s kite drawing.
Let $c$ denote the crossing point of the kite's edges $e$ and $f$.
Then, the interior angle at $c$ must be less than $\pi$, which
implies a bend at $c$ for both $e$ and $f$,
a contradiction to $\Drawing{G}$ being straight-line.
\Qed
\end{proof}
So far, we conclude that the induced embedding $\Embedding{G}$
must be as depicted in \Fig{RAC} if there is a \RAC{} drawing of $G^+$.
However, this embedding is not realizable with right angle crossings.
\begin{lemma}
  Graph $G^+$ does not admit a \RAC{} drawing.
\end{lemma}
\begin{proof}
  Assume that $G^+$ has a \RAC{} drawing $\Drawing{G^+}$. By \Lem{outer},
  the outer face of the induced drawing $\Drawing{G}$ is a \Trivial{} triangle $\Triangle$.
  Every fat edge of $G$ is a planar edge of a $K_4$ in $\Drawing{G}$, or more
  specifically of a kite (\Lem{convex}). Hence, every edge bounding $\Triangle$ is
  shared with a kite $\Kfour_i$ for $i = 1, 2, 3$, which is located inside
  $\Kthree$. Let $c_i$ denote the point in the plane where the two edges of
  $\Kfour_i$ cross each other and let $\Kthree_i$ be the non-\Trivial{}
  triangle of the kite embedding of $\Kfour_i$ that is bounded by one of the
  edges of $\Triangle$. Then, the interior angle of $\Kthree_i$ at $c_i$
  must be $\frac{\pi}{2}$ and subsequently, the two remaining interior
  angles of $\Kthree_i$ sum up to $\pi - \frac{\pi}{2} = \frac{\pi}{2}$.
  Observe that the kites' faces are pairwisely disjoint by \Lem{convex}.
  Hence, the sum of $\Triangle$'s interior angles must be strictly greater
  than $\frac{3\pi}{2}$, a contradiction to $\Triangle$ being a triangle.
  \Qed
\end{proof}

\section{Recognition}
\SectLabel{recognition}
The recognition problem for \One{}-planar and \IC{}-planar graphs is
\NP-complete, even if the graphs are $3$-connected and are given with a
rotation system \cite{GB-AGEFCE-07, abgr-1prs-15, bdeklm-IC-16}.
However, triangulated, maximal, and optimal \One{}-planar
graphs can be recognized in time $\bigO(n^3)$ \cite{cgp-rh4mg-06},
$\bigO(n^5)$ \cite{b-4mg1p-15}, and $\bigO(n)$ \cite{b-optlin-16}, respectively.

\begin{figure}[tb]
\centering
\includegraphics{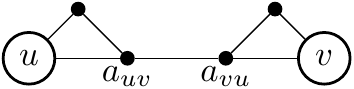}%
\caption{Gadget replacing every edge $\Edge{u}{v}$ in the \NP{}-reduction.}
\FigLabel{np-gadget}
\end{figure}

Our \NP-hardness result solves an open problem by Zhang \cite{z-dcmgprc-14}
and can be obtained from known \NP-hardness proofs 
\cite{bdeklm-IC-16,GB-AGEFCE-07}, 
\eg, by a reduction from \One{}-planarity as in \cite{bdeklm-IC-16}, which
replaces every edge $\Edge{u}{v}$ of a graph $G = (V,E)$ by the gadget in
\Fig{np-gadget}.
Then, in every \IC{}-planar and even \NIC{}-planar embedding $\Embedding{G'}$ of
the resulting graph $G'$, every crossed edge must be an edge $\Edge{a_{uv}}{a_{vu}}$
for some $\Edge{u}{v} \in E$ and
$\Embedding{G'}$ exists if and only if the induced embedding of $G$ is \One{}-planar.
As testing \One{}-planarity is \NP{}-complete, we obtain:
\begin{corollary}
It is \NP-complete to test whether a graph is \NIC{}-planar.
\end{corollary}
Recently, Brandenburg \cite{b-ripnp-16} has used the relationship
between \One{}-planar graphs and hole-free 4-map graphs as in
\Cor{triang-vs-4map} and the cubic-time recognition algorithm for
(hole-free) 4-map graphs of Chen \ea{} \cite{cgp-rh4mg-06} to develop a
cubic-time recognition algorithm for triangulated \NIC{}-planar
(\IC{}-planar) graphs from which he obtained an $\bigO(n^5)$ time algorithm
for maximal and a cubic-time algorithm for densest \NIC{}-planar
(\IC{}-planar) graphs.

\section{Conclusion}
\label{sect:conclusion}
For a natural subclass of \One{}-planar graphs, we presented
diverging, yet tight upper and lower bounds for maximal graphs.
Paralleling the result that there are maximal \One{}-planar graphs that are
sparser than maximal planar graphs, we showed that there are maximal \NIC{}-planar
graphs that are sparser than maximal \IC{}-planar graphs.
Our tool is a generalized dual graph and a condensation of $K_4$ subgraphs.
Whereas \IC{}-planar graphs are a subset of \RAC{} graphs, we showed that
\NIC{}-planar graphs and \RAC{} graphs are incomparable.
The proof of \Thm{notrac} shows, to the best of our knowledge for
the first time, that there are non-\RAC{}, \One{}-planar graphs with a
density less than the upper bound for \RAC{} graphs of $4n - 10$.
Finally, we showed that the recognition of \NIC{}-planar graphs is
\NP{}-hard in general, whereas optimal \NIC{}-planar graphs  can be
recognized in linear time.

Future work are similar characterizations for \IC{}-planar graphs in terms of
generalized duals and the linear-time recognition of optimal \IC{}-planar
graphs.

\bibliographystyle{splncs03}
\bibliography{paper}
\end{document}